\documentclass[runningheads,oribibl]{llncs}
\usepackage[T1]{fontenc}
\usepackage[numbers,sort&compress]{natbib}

\makeatletter
\renewcommand\@biblabel[1]{#1.}
\makeatother

\usepackage[colorlinks]{hyperref}
\usepackage{color}

\usepackage{bm}
\usepackage{amsmath}
\usepackage{mathtools}
\usepackage{hyperref}
\usepackage{cleveref}

\usepackage{caption}
\usepackage{booktabs}
\usepackage{multirow}

\newcommand{\real}{\mathbb{R}}

\newcommand{\integer}{\mathbb{Z}}

\newcommand{\phenom}[0]{local suboptimality}
\newcommand{\Phenom}[0]{Local suboptimality}
\newcommand{\PheNom}[0]{Local Suboptimality}
\newcommand{\Nash}[0]{\textrm{Nash}}

\newcommand{\V}[1]{\bm{#1}}

\newcommand{\vsigma}{\V{\sigma}}
\newcommand{\va}{\V{a}}
\newcommand{\vc}{\V{c}}
\newcommand{\vmu}{\V{\mu}}
\newcommand{\gls}{\mathcal{G}_{LS}}
\newcommand{\glo}{\mathcal{G}_{LO}}
\newcommand{\glspp}{\mathcal{G}_{LS}^{p,p}}
\newcommand{\glopp}{\mathcal{G}_{LO}^{p,p}}
\newcommand{\glsmm}{\mathcal{G}_{LS}^{m,m}}
\newcommand{\glomm}{\mathcal{G}_{LO}^{m,m}}
\newcommand{\glsmp}{\mathcal{G}_{LS}^{m,p}}
\newcommand{\glomp}{\mathcal{G}_{LO}^{m,p}}
\newcommand{\spe}[0]{\textrm{SPE}}

\newcommand{\vsmin}{\vsigma^{\min}}
\newcommand{\vsmax}{\vsigma^{\max}}

\newcommand{\condSet}[2]{\left\{#1 \;\middle|\; #2\right\}}

\usepackage{tikz}
\usetikzlibrary{shapes.misc}
\usetikzlibrary{arrows.meta}
\usetikzlibrary{calc}
\usetikzlibrary{decorations.pathreplacing}
\tikzset{cross/.style={cross out, draw=black, minimum size=2*(#1-\pgflinewidth),
    inner sep=0pt, outer sep=0pt},
cross/.default={1pt}
}

\usepackage{algorithm}
\usepackage{algpseudocode}
\algrenewcommand\algorithmicrequire{\textbf{Input:}}
\algrenewcommand\algorithmicensure{\textbf{Output:}}
\algnewcommand{\algorithmicand}{\textbf{ and }}
\algnewcommand{\algorithmicor}{\textbf{ or }}
\algnewcommand{\OR}{\algorithmicor}
\algnewcommand{\AND}{\algorithmicand}

\crefname{question}{Question}{Questions}

\usepackage{amsfonts}
\usepackage{amssymb}

\begin{document}
\title{Emergence of Locally Suboptimal Behavior in Finitely Repeated Games}
%\author{First Author\inst{1}\orcidID{0000-1111-2222-3333} \and
%Second Author\inst{2,3}\orcidID{1111-2222-3333-4444} \and
%Third Author\inst{3}\orcidID{2222--3333-4444-5555}}
\author{Yichen Yang \and
Martin Rinard
}
\authorrunning{Yichen Yang \and
Martin Rinard
}
\institute{Department of Electrical Engineering and Computer Science, \\
Massachusetts Institute of Technology, USA \\
\email{\{yicheny,rinard\}@csail.mit.edu}}
%\institute{Princeton University, Princeton NJ 08544, USA \and
%Springer Heidelberg, Tiergartenstr. 17, 69121 Heidelberg, Germany
%\institute{Anonymous Institute}
%\author{Anonymous authors}
\maketitle              % typeset the header of the contribution

\begin{abstract}
We study the emergence of locally suboptimal behavior in finitely repeated games. Locally suboptimal behavior refers to players play suboptimally in some rounds of the repeated game (i.e., not maximizing their payoffs in those rounds) while maximizing their total payoffs in the whole repeated game. The central research question we aim to answer is when locally suboptimal behavior can arise from rational play in finitely repeated games. 
In this research, we focus on the emergence of locally suboptimal behavior in subgame-perfect equilibria (SPE) of finitely repeated games with complete information.
We prove the first sufficient and necessary condition on the stage game $G$ that ensure that, for all $T$ and all subgame-perfect equilibria of the repeated game $G(T)$, the strategy profile at every round of $G(T)$ forms a Nash equilibrium of the stage game $G$.
We prove the sufficient and necessary conditions for three cases: 1) only pure strategies are allowed, 2) the general case where mixed strategies are allowed, and 3) one player can only use pure strategies and the other player can use mixed strategies. 
Based on these results, we obtain complete characterizations on when allowing players to play mixed strategies will change whether \phenom{} can ever occur in some repeated game.
Furthermore, we present an algorithm for the computational problem of, given an arbitrary stage game, deciding if locally suboptimal behavior can arise in the corresponding finitely repeated games. This addresses the practical side of the research question.
\end{abstract}

% vim: tw=0 wrap linebreak filetype=tex foldmethod=marker foldmarker=f{{{,f}}} spell spelllang=en

\section{Introduction}

We study the emergence of locally suboptimal behavior in finitely repeated games. Locally suboptimal behavior refers to players play suboptimally in some rounds of the repeated game (i.e., not maximizing their payoffs in those rounds) while maximizing their total payoffs in the whole repeated game. The emergence of locally suboptimal behavior reflects some fundamental psychological and social phenomena, such as delayed gratification, threats, and enforced cooperation. 

We focus on the emergence of locally suboptimal behavior in subgame-perfect equilibria (SPE) of finitely repeated games with complete information.
A widely known result that appears in many textbooks and lecture notes is that if the stage game $G$ has a unique Nash equilibrium payoff for every player, then in any SPE of any finitely repeated game $G(T)$ with any $T$ rounds, the strategy profile at each round forms an NE of the stage game $G$ \cite{fudenberg1991game,gibbons1992primer,osborne1994course}. This is proved using backward induction. 
It is also known that there are stage games $G$ where in some SPEs of the repeated game $G(T)$ for some $T$, the strategy profile at some round does not form a stage-game NE (\Cref{sec:example-off-nash} presents an example).
Such off-(stage-game)-Nash play occurs due to {\em `threats'} between players that are stated implicitly through players' strategies.

We define such off-(stage-game)-Nash plays in repeated games as {\em \phenom{}}.
As we have seen, for some stage games, \phenom{} can occur in some SPE of some repeated games; for other stage games, \phenom{} can never occur in any SPE of any repeated games.
Therefore, we can partition the set of all stage games $\mathcal{G}$ into two disjoint subsets $\gls$ and $\glo$. $\gls$ is the set of stage games $G$ where \phenom{} occurs in some SPE of $G(T)$ for some $T$; $\glo$ is the set of stage games $G$ where \phenom{} never occurs in any SPE of $G(T)$ for any $T$ (LO stands for locally optimal).
Our goal in this research is to completely characterize which stage games belong to $\gls$ and $\glo$. The central research question we aim to tackle is:

\begin{question}
\label{ques:cond}
What is a sufficient and necessary condition on the stage game $G$ that ensures that, for all $T$ and all subgame-perfect equilibria of the repeated game $G(T)$, the strategy profile at every round of $G(T)$ forms a Nash equilibrium of the stage game $G$?
\end{question}

The answer to \Cref{ques:cond} completely characterizes $\gls$ and $\glo$. As we have discussed, a sufficient condition for \Cref{ques:cond} is widely known (uniqueness of Nash equilibrium payoff for each player). However, this condition is not necessary; in fact, no previous work establishes a sufficient and necessary condition. A large body of work focuses on Folk Theorems \cite{Benoit1985,benoit1987,gossner1995,Smith1995,gonzalez2006finitely}, where the property of interest is: all feasible (i.e., the payoff profile lies in the convex hull of the set of all possible payoff profiles of the stage game) and individually rational (i.e., the payoff of each player is at least their minmax payoff in the stage game) payoff profiles can be attained in the equilibrium of the repeated game. As we show in \Cref{sec:example-diff}, the property considered in Folk Theorems and the \phenom{} property we consider in this work are different, and the two properties do not have direct implications in either direction. Therefore, the conditions established for Folk Theorems in the literature do not solve the problem we consider.

In addition to the complete mathematical characterization of the partitioning between $\gls$ and $\glo$, we also consider the computational aspect of the problem: 

\begin{question}
\label{ques:comp}
    Given an arbitrary stage game $G$, how to (algorithmically) decide if there exists some $T$ and some SPE of $G(T)$ where \phenom{} occurs? Is this problem decidable?
\end{question}

A naive approach is to enumerate over $T$, solve for all SPEs for each $G(T)$, and check if off-Nash behavior occurs. Such an approach is not only computationally inefficient, but also not guaranteed to terminate due to the unboundedness of $T$. In fact, we show that there are stage games where \phenom{} only occurs in repeated games with very large $T$, and we can construct games where this minimum $T$ for \phenom{} to occur can be arbitrarily large (\Cref{sec:example-largeT}). These facts motivate the study of \Cref{ques:comp}.

\subsection{Summary of Results}

\subsubsection{Sufficient and Necessary Conditions for 2-Player Games}

A main theoretical contribution of this part is that we prove sufficient and necessary conditions for \Cref{ques:cond} for 2-player games. We prove the conditions for three cases: 1) only pure strategies are allowed (\Cref{thm:pure}), 2) the general case where mixed strategies are allowed (\Cref{thm:general}), and 3) one player can only use pure strategies and the other player can use mixed strategies (\Cref{thm:pure-vs-mixed}). 
This is the first sufficient and necessary condition for off-(stage-game)-Nash plays to occur in SPEs of 2-player finitely repeated games.

From the perspective of partitioning the set of stage games $\mathcal{G}$, denote $\glspp$ as the set of stage games $G$ where \phenom{} occurs in some SPE of $G(T)$ for some $T$ when both players can only use pure strategies, $\glopp$ as the set of stage games $G$ where \phenom{} never occurs in any SPE of $G(T)$ for any $T$ when both players can only use pure strategies, $\glsmm$ and $\glomm$ as the corresponding partitioning when both players can use mixed strategies, and $\glsmp$ and $\glomp$ as the corresponding partitioning when player 1 can use mixed strategies and player 2 can only use pure strategies.
Essentially, we obtain complete mathematical characterizations of the partitioning of $\mathcal{G}$ for cases (1), (2), and (3) above: 1) $\glspp$ and $\glopp$, 2) $\glsmm$ and $\glomm$, and 3) $\glsmp$ and $\glomp$.

\subsubsection{Effect of Changing from Pure Strategies to Mixed Strategies on the Emergence of Local Suboptimality}

Based on the above results, we further study the effect of changing from pure strategies to mixed strategies on the emergence of \phenom{}. We aim to answer the following question: under what conditions on the stage game $G$ will allowing players to play mixed strategies change whether \phenom{} can ever occur in some repeated game $G(T)$? Essentially, we aim to study the relationships between $\glspp$, $\glsmp$, and $\glsmm$.

We prove that $\glspp \subseteq \glsmp \subseteq \glsmm$ (\Cref{thm:pp-to-mp-1,thm:mp-to-mm-1,thm:pp-to-mm-1}), i.e., if \phenom{} can occur before the change, then after changing any player (or both players) from pure-strategies-only to mixed-strategies-allowed, \phenom{} can still occur. This is because we prove that any SPE of the repeated game before the change is still an SPE after the change, and any strategy profile that is not a stage-game NE before the change is still not a stage-game NE after the change.

On the other hand, we show that $\glspp \neq \glsmp$ and $\glsmp \neq \glsmm$ (so $\glspp$ is a proper subset of $\glsmp$ and $\glsmp$ is a proper subset of $\glsmm$), i.e., there are games where \phenom{} can never occur before the change, but after changing one player (or both players) from pure-strategies-only to mixed-strategies-allowed, \phenom{} can occur. We present complete characterizations of the sets $\glsmp \setminus \glspp$, $\glsmm \setminus \glsmp$, and $\glsmm \setminus \glspp$, by proving sufficient and necessary conditions on the stage game $G$ such that \phenom{} can never occur before the change but can occur after the change (\Cref{thm:pp-to-mp-2,thm:mp-to-mm-2,thm:pp-to-mm-2}). Our characterizations are fine-grained based on $|V_1|$ and $|V_2|$, the number of payoff values attainable at stage-game NEs for each player. For example, we show that under certain preconditions on $|V_1|$ and $|V_2|$, $\glspp = \glsmp$; under other preconditions on $|V_1|$ and $|V_2|$, $\glspp \neq \glsmp$, and for each of such cases, we present an example stage game $G$ where $G\notin \glspp$ and $G\in \glsmp$. These examples demonstrate different mechanisms of how changing a player from pure-strategies-only to mixed-strategies-allowed can lead to the emergence of \phenom{}. We perform the same fine-grained analyses on $\glsmm \setminus \glsmp$ and $\glsmm \setminus \glspp$ as well.

\subsubsection{Computational Aspects}

We propose an algorithm for deciding \Cref{ques:comp} for 2-player games for the general case where mixed strategies are allowed and analyze the computational complexity of this algorithm. This shows that \Cref{ques:comp} is decidable for 2-player games where mixed strategies are allowed. This algorithm provides a method for computationally deciding if \phenom{} can ever happen for a given stage game. The algorithm is based on the sufficient and necessary condition established in \Cref{thm:general}. We design several efficient methods for checking different parts of the condition by utilizing properties we prove for general games. Naive methods for checking these parts of the condition take exponential time in the worst case, whereas our methods for checking these parts of the condition take polynomial time in the worst case.

\subsubsection{Generalization to $n$-Player Games}

We prove a separate sufficient condition and necessary condition for \Cref{ques:cond} for $n$-player games. These conditions are both tighter than what is previously known in the literature (again, only a sufficient condition is known previously, i.e., there is a unique Nash equilibrium payoff for each player \cite{fudenberg1991game,gibbons1992primer,osborne1994course}). 

The proof of a sufficient and necessary condition for the 2-player case relies on some properties that we prove to hold for 2-player games (\Cref{lemma:matrix} and the subsequent arguments in the proof of \Cref{thm:general} that uses \Cref{lemma:matrix} to show there exists a connected component of off-Nash strategy profiles). It is not clear whether similar properties hold for $n$-player games. Therefore, the questions of 1) what is a sufficient and necessary condition for $n$-player games, and 2) is \Cref{ques:comp} decidable for $n$-player games, remain open.

\section{The Model}
\label{sec:model}

A stage game $G = (n, A_1, \dots, A_n, u_1, \dots, u_n)$ in normal form consists of $n$ players, each player $i$'s strategy space $A_i$, and each player $i$'s payoff function $u_i: A \rightarrow \real$, where $A = A_1 \times A_2 \times \dots \times A_n$. We assume $n$ and $A$ are finite. Throughout this chapter, we use $a$ to denote pure strategies (or actions) in the stage game and $\sigma$ to denote mixed strategies in the stage game, e.g. $a_i\in A_i$ denotes a pure strategy for player $i$ and $\sigma_i \in \Delta A_i$ denotes a mixed strategy for player $i$, both for the stage game, where $\Delta S$ denotes the set of probability distributions over set $S$. We use $S_{\sigma_i} = \condSet{a}{a\in A_i, \sigma_i(a) > 0}$ to denote the support for mixed strategy $\sigma_i$. A strategy profile $\vsigma = (\sigma_1, \dots, \sigma_n)$ is a set of strategies for all players. In general, we use bold symbols to represent collections over players. For convenience, we use $u_i(\vsigma)$ to denote the expected payoff of player $i$ under the (mixed) strategy profile $\vsigma$. A strategy $\sigma_i$ is a \textit{best response} to the strategy profile of the other players $\vsigma_{-i}$ if $u_i(\sigma_i, \vsigma_{-i}) = \max_{\sigma_i' \in \Delta A_i} u_i(\sigma_i', \vsigma_{-i})$. A strategy profile $\vsigma$ is a \textit{Nash equilibrium} (NE) if for all player $i\in[n]$ ($[n]$ denotes the set $\{1,\dots, n\}$), $\sigma_i$ is a best response to $\vsigma_{-i}$. We use $\textrm{Nash}(G)$ to denote the set of all Nash equilibria of the stage game $G$. We use $V_i = \condSet{u_i(\vsigma)}{\vsigma\in\Nash(G)}$ to denote the set of payoff values attainable at Nash equilibria for player $i$.

We use $G(T)$ to denote the game where $G$ is played repeatedly for $T$ rounds, where $T$ is a positive integer. Denote the \textit{outcome} in round $t\in[T]$ as $\va^t \in A$. Player $i$'s total payoff in the repeated game $G(T)$ is $U_i = \sum_{t=1}^T u_i(\va^t)$. A strategy of player $i$ in $G(T)$ specifies which actions to take in each round given any history of play in the previous rounds. Formally, denote a \textit{history} of play in the first $k$ rounds as $h(k) = (\va^1, \dots, \va^k)$, and the set of all possible $k$-round histories as $H(k) = A^k$ ($H(0)$ denotes the singleton set containing the empty history). A (mixed) strategy of player $i$ in $G(T)$ can be represented as $\mu_i: H \rightarrow \Delta A_i$, where $H = \cup_{k=0}^{T-1} H(k)$ is the set of all histories. This form of representation is also commonly known as \textit{behavior strategies}. We use $\vmu = (\mu_1,\dots,\mu_n)$ to denote strategy profiles of $G(T)$, and the concept of best response and Nash equilibrium are defined in the same way as for the stage game. 

In this paper, we focus on \textit{subgame-perfect equilibria} (SPE) of $G(T)$. SPE was originally introduced by \cite{Selten1965,Selten1975} to eliminate NEs that involve non-credible threats off the equilibrium path. Given a strategy $\mu_i$ of player $i$ for $G(T)$, denote $\mu_{i|h(k)}$ as the resulting strategy for subgame $G(T-k)$ obtained by conditioning $\mu_i$ on some history $h(k)$.
Formally, given $h(k) = (\va^1, \dots, \va^k)$, $\mu_{i|h(k)}$ is given by: 1) $\mu_{i|h(k)}(h(0)) = \mu_i(\va^1, \dots, \va^k)$; 2) for any $t<T-k$ and any $(\vc^1, \vc^2, \dots, \vc^t)\in H(t)$, $\mu_{i|h(k)}(\vc^1, \vc^2, \dots, \vc^t) = \mu_i(\va^1, \dots, \va^k, \vc^1, \vc^2, \dots, \vc^t)$.
And denote $\vmu_{|h(k)} = (\mu_{1|h(k)}, \dots, \mu_{n|h(k)})$. A strategy profile $\vmu$ is an SPE if for all $0\leq k < T$ and all $h(k)\in H(k)$, $\vmu_{|h(k)}$ is an NE of $G(T-k)$. We use $\spe(G,T)$ to denote the set of all SPEs of the repeated game $G(T)$.

The phenomenon we are interested in is when in some SPE of the repeated game $G(T)$, the behavior strategy profile in some round does not form an NE of the stage game $G$. In other words, some player uses a locally suboptimal strategy in some round, in the sense that the strategy is not a best response for that round, as part of an SPE in the repeated game. We formally define this phenomenon of \phenom{} as follows.
\begin{definition}[\Phenom{}]
\label{def:local-sub}
\Phenom{} occurs in some SPE $\vmu$ of some repeated game $G(T)$ if there exists some $0\leq k < T$ and play history $h(k)\in H(k)$ where $\vmu(h(k)) = \Big(\mu_1(h(k)), \dots, \mu_n(h(k))\Big) \notin \textrm{Nash}(G)$, i.e. the behavior strategy profile at some round does not form an NE of the stage game. 
\end{definition}
We refer to such behavior strategy profiles that do not form an NE of the stage game as {\em off-(stage-game)-Nash} plays, or {\em off-Nash} plays in short.

Denote the set of all stage games $G$ as $\mathcal{G}$ ($\mathcal{G}$ is an infinite set). $\mathcal{G}$ can be partitioned into two disjoint subsets $\gls$ and $\glo$. $\gls$ is the set of stage games $G$ where \phenom{} occurs in some SPE of $G(T)$ for some $T$; $\glo$ is the set of stage games $G$ where \phenom{} never occurs in any SPE of $G(T)$ for any $T$ (LO stands for locally optimal). Our central research questions stated in \Cref{ques:cond} and \Cref{ques:comp} are essentially about solving the following problems: 1) completely characterize $\gls$ (thus $\glo$) using mathematical conditions, and 2) given any stage game $G$, algorithmically determine if $G$ is in $\gls$ or $\glo$.

\section{Motivating Examples}
\label{sec:suboptimal-example}

In this section, we present several example games that motivate our study on the emergence of \phenom{} in finitely repeated games. 

\subsection{Example Game where \PheNom{} Occurs}
\label{sec:example-off-nash}

We first present a simple example game where \phenom{} occurs to give a flavor of how such phenomena arise.

\begin{example}
\label{example:off-nash}

\begin{table}[ht]
    \centering
    \begin{tabular}{c|c|c|}
    & $a_2$ & $b_2$ \\
    \midrule
    $a_1$ & (3,1) & (0,1) \\
       \midrule
   $b_1$ & (2,1) & (1,1) \\
   \midrule
    \end{tabular}
    \caption{Example stage game $G$ where \phenom{} occurs in an SPE of $G(2)$. }
    \label{tab:off_nash}
\end{table}

\Cref{tab:off_nash} presents an example stage game $G$ where \phenom{} occurs in an SPE of the repeated game $G(2)$. The game is represented in matrix form. Row player chooses from actions $a_1$ and $b_1$, column player chooses from actions $a_2$ and $b_2$. In each entry of the matrix, the first value is the payoff of the row player, and the second value is the payoff of the column player. 

The strategy profile $(b_1, a_2)$ is not an NE of the stage game $G$. However, the following is an SPE of the 2-round repeated game $G(2)$, in which the strategy profile in the first round is $(b_1, a_2)$:
\begin{itemize}
    \item In the second round, if the row player plays $a_1$ in the first round, play $(b_1, b_2)$; else, play $(a_1, a_2)$.
    \item In the first round, play $(b_1, a_2)$.
\end{itemize}
Notice that although the row player can obtain an addition payoff of 1 in the first round by switching to play $a_1$ in the first round, they will lose a payoff of 2 in the second round. This is why the above strategy profile is an SPE of the repeated game. Intuitively, the column player `threatens' the row player by stating (implicitly through the column player's strategy) that if the row player deviates in the first round, the column player will play according to the stage-game NE that gives a lower payoff to the row player in the second round.
\end{example}

\subsection{Example Game where SPE with \PheNom{} Strictly Dominates SPEs without \PheNom{}}

Here we present an example game where in the repeated game, some SPE in which \phenom{} occurs strictly dominates all SPEs where \phenom{} does not occur. 

\begin{example}

\begin{table}[ht]
    \centering
    \begin{tabular}{c|c|c|c|}
    & $a_2$ & $b_2$ & $c_2$ \\
    \midrule
    $a_1$ & (3,3) & (0,4) & (0,0) \\
      \midrule
     $b_1$ & (4,0) & (2,2) & (0,1) \\
       \midrule
   $c_1$ & (0,0) & (1,0) & (1,1) \\
   \midrule
    \end{tabular}
    \caption{Example stage game where in the repeated game, some SPE in which \phenom{} occurs strictly dominates all SPEs where \phenom{} does not occur.}
    \label{tab:off_nash_dom}
\end{table}

\Cref{tab:off_nash_dom} presents the example stage game in matrix form. This stage game $G$ has three Nash equilibria: $(b_1, b_2)$, $(c_1, c_2)$, and a mixed NE $(\sigma_1, \sigma_2)$ where $\sigma_1(b_1) = \sigma_1(c_1) = \sigma_2(b_2) = \sigma_2(c_2) = 0.5$. The payoffs of each of the above NEs are: $(2,2)$, $(1,1)$, and $(1,1)$ respectively. Therefore, for a $T$-round repeated game $G(T)$, in any SPE where \phenom{} does not occur, the total payoff of each player is at most $2T$. We argue that for any $T>2$, the following strategy profile is an SPE of $G(T)$:
\begin{itemize}
    \item In the first $T-2$ rounds, the row player plays $a_1$ and the column player plays $a_2$ (note that this strategy profile is not an NE of the stage game $G$). If any player deviates to other actions in any round, the two players immediately switch to play $(c_1, c_2)$ for the rest of the game.
    \item In the last 2 rounds, players play $(b_1, b_2)$.
\end{itemize}

The total payoff of each player under the above SPE is $3T-2$. For all $T>2$, $3T-2 > 2T$. Therefore, the above SPE in which \phenom{} occurs strictly dominates any SPE in which \phenom{} does not occur.
\end{example}

\subsection{Example Games Demonstrating Difference Between \PheNom{} and the Property in Folk Theorems}
\label{sec:example-diff}

Under the theme of analyzing equilibrium solutions in repeated games, a large body of work focuses on Folk Theorems \cite{Benoit1985,benoit1987,gossner1995,Smith1995,gonzalez2006finitely}, where the property of interest is: all {\em feasible} (i.e., the payoff profile lies in the convex hull of the set of all possible payoff profiles of the stage game) and {\em individually rational} (i.e., the payoff of each player is at least their minmax payoff in the stage game) payoff profiles can be attained in equilibria of the repeated game. Here we show that the property considered in Folk Theorems and the \phenom{} property we consider in this research are different, and the two properties do not have direct implications in either direction. We present 1) an example game where  \phenom{} can occur, but not all feasible and individually rational payoffs can be attained in the repeated game, and 2) an example game where all feasible and individually rational payoffs can be attained in the repeated game, but \phenom{} cannot occur.

\begin{example}
\Cref{tab:example-no-folk} presents an example stage game $G$ where \phenom{} can occur in the repeated game, but not all feasible and individually rational payoffs can be attained in the repeated game. This example is taken from \cite{Benoit1985}. This game contains 3 players. Player 1 selects rows ($a_1, b_1, c_1$), player 2 selects columns ($a_2, b_2$), and player 3 selects matrices ($a_3, b_3$). While \cite{Benoit1985} analyzes this example with only pure strategies allowed, we consider the general case where mixed strategies are allowed. There are three Nash equilibria: (i) $(a_1, a_2, a_3)$; (ii) $(a_1, b_2, b_3)$; (iii) $(a_1, \sigma_2, \sigma_3)$ where $\sigma_2(a_2) = \sigma_2(b_2) = 0.5$, $\sigma_3(a_3) = 0.25$, $\sigma_3(b_3) = 0.75$. These equilibria achieve payoffs of $(3,3,3)$, $(2,2,2)$, $(1.5,1.5,1.5)$ respectively. Following a similar idea in \Cref{example:off-nash}, it is easy to construct an SPE in a repeated game where \phenom{} occurs. For example, the following strategy profile is an SPE of $G(4)$:
\begin{itemize}
    \item In the first round, play $(a_1, a_2, b_3)$.
    \item In the last three rounds, if players play in the first round is $(a_1, a_2, b_3)$, play $(a_1, a_2, a_3)$; otherwise, play $(a_1, b_2, b_3)$.
\end{itemize}
In this SPE, the first round play does not form a stage-game NE. Therefore, \phenom{} occurs.

\begin{table}[ht]
    \centering
    \begin{tabular}{c|c|c|}
     & $a_2$ & $b_2$ \\
    \midrule
     $a_1$ & (3,3,3) & (0,0,0) \\
    \midrule
    $b_1$ & (0,0,0) & (0,0,0) \\
    \midrule
    $c_1$ & (0,1,1) & (0,0,0) \\
   \midrule
   \multicolumn{3}{c}{$a_3$}
    \end{tabular}
    \quad
    \begin{tabular}{c|c|c|}
     & $a_2$ & $b_2$ \\
    \midrule
     $a_1$ & (1,1,1) & (2,2,2) \\
    \midrule
    $b_1$ & (0,1,1) & (0,1,1) \\
    \midrule
    $c_1$ & (0,1,1) & (0,0,0) \\
   \midrule
   \multicolumn{3}{c}{$b_3$}
    \end{tabular}
    \caption{Example stage game where \phenom{} can occur, but not all feasible and individually rational payoffs can be attained in the repeated game. }
    \label{tab:example-no-folk}
\end{table}

For this stage game $G$, each player's minmax payoff is 0. We follow a similar argument as \cite{Benoit1985}. Denote $w_i(T)$ as the worst payoff that player $i$ can get in any SPE of $G(T)$, the $T$-round repeated game. We claim that for $i=2,3$, $w_i(T)/T\geq 0.5$, therefore not all feasible and individually rational payoffs can be approximated. We use induction. The claim is true for $T=1$. Suppose $w_i(T-1)\geq 0.5(T-1)$. Consider the strategy profile $\vmu$ in $G(T)$ that attains $w_2(T)$ and $w_3(T)$. Notice that player 2 and 3 always get the same payoff in this game, so $w_2(T)$ and $w_3(T)$ will be attained at the same time. Consider the behavior strategy profile in the first round $\vsigma = (\sigma_1, \sigma_2, \sigma_3)$ as specified in $\vmu$. If $\sigma_1(c_1)\cdot \sigma_2(b_2) \leq 0.5$, then player 3 playing $b_3$ in the first round gives them at least a total payoff of $0.5 + w_3(T-1)$. This implies $w_3(T)\geq 0.5 + w_3(T-1)$ and we are done by the induction hypothesis. If $\sigma_1(c_1)\cdot \sigma_2(b_2) > 0.5$, then player 2 playing $a_2$ in the first round gives them at least a total payoff of $0.5+w_2(T-1)$. This implies $w_2(T)\geq 0.5 + w_2(T-1)$ and we are done by the induction hypothesis.
\end{example}

\begin{example}
\Cref{tab:every_ne} presents an example stage game $G$ where all feasible and individually rational payoffs can be attained in the repeated game, but \phenom{} cannot occur. For this stage game $G$, the set of all feasible and individually rational payoff profiles is $\condSet{(u_1, u_2)}{0\leq u_1 \leq 1, 0\leq u_2 \leq 1}$. All these feasible and individually rational payoffs can be attained in the stage game $G$ itself, which is also a one-round repeated game $G(1)$. But every possible strategy profile in $G$ is an NE of $G$, so \phenom{} cannot occur.

\begin{table}[ht]
    \centering
    \begin{tabular}{c|c|c|}
    & $a_2$ & $b_2$ \\
    \midrule
    $a_1$ & (0,0)   & (1,0) \\
    \midrule
    $b_1$ & (0,1)  & (1,1) \\
   \midrule
    \end{tabular}
    \caption{Example stage game where every possible strategy profile is an NE.}
    \label{tab:every_ne}
\end{table}

\end{example}

\subsection{Example Game where \PheNom{} Only Occurs with Large $T$}
\label{sec:example-largeT}

One of the research questions we consider is in the computational aspect: given an arbitrary stage game $G$, how to (algorithmically) decide if there exists some $T$ and some SPE of $G(T)$ where \phenom{} occurs? A naive approach is to enumerate over $T$, solve for all SPEs for each $G(T)$, and check if off-(stage-game)-Nash behavior occurs. Here, we present a construction of games where \phenom{} only occurs with arbitrarily large $T$. This means that the naive approach above might need to check an arbitrarily large number of $T$'s before returning a result.

\begin{example}

\begin{table}[ht]
    \centering
    \begin{tabular}{c|c|c|}
    & $a_2$ & $b_2$ \\
    \midrule
    $a_1$ & (3,2)   & ($\alpha$,1) \\
      \midrule
     $b_1$ & (3,2)   & (2,2) \\
      \midrule
      $c_1$ & ($\alpha$,1)  & (2,2) \\
       \midrule
    \end{tabular}
    \caption{Example stage game where \phenom{} only occurs with large $T$.}
    \label{tab:large_T}
\end{table}

\Cref{tab:large_T} presents a construction of stage games where \phenom{} only occurs with arbitrarily large $T$. We claim that for any $\alpha < 2$, \phenom{} cannot occur with any $T < \frac{1}{2}(2-\alpha)$, and \phenom{} can occur with any $T > 3 - \alpha$. Therefore, as $\alpha$ becomes smaller, we have games where \phenom{} only occurs with arbitrarily large $T$. We present the proof as follows.

It is easy to see that the set of Nash equilibria of this stage game $G$ is
\begin{itemize}
    \item $(\sigma_1, a_2)$ where $\sigma_1(a_1) = \lambda, \sigma_1(b_1) = 1 - \lambda$ for all $0\leq \lambda \leq 1$,
    \item $(b_1, \sigma_2)$ where $\sigma_2(a_2) = \lambda, \sigma_2(b_2) = 1 - \lambda$ for all $0\leq \lambda \leq 1$,
    \item $(\sigma_1, b_2)$ where $\sigma_1(b_1) = \lambda, \sigma_1(c_1) = 1 - \lambda$ for all $0\leq \lambda \leq 1$.
\end{itemize}
It follows that $V_1 = [2,3]$, the continuous range from 2 to 3, and $V_2 = \{2\}$. 

First, we show that for any $T > 3 - \alpha$, \phenom{} can occur. Here is an SPE where \phenom{} occurs:
\begin{itemize}
    \item The first round strategy profile is $(\hat{\sigma}_1, b_2)$ where $\hat{\sigma}_1(a_1) = \frac{1}{4}$ and $\hat{\sigma}_1(c_1) = \frac{3}{4}$. $(\hat{\sigma}_1, b_2)\notin \Nash(G)$ since player 1 is not playing a best response (but player 2 is playing a best response).
    \item If player 1's first round play is $b_1$ or $c_1$, we let the players play a stage game Nash equilibrium that achieves $u_1=2$ (minimum payoff for player 1) in all the remaining $T-1$ rounds. 
    \item If player 1's first round play is $a_1$, we let the players play a sequence of stage game Nash equilibria that achieves a total payoff $U_1 = 2(T-1) + 2 - \alpha$ in the remaining $T-1$ rounds. This is possible since $T-1 > 2 - \alpha$ and $V_1$ contains the continuous interval between 2 and 3.
\end{itemize}

Now let $T^*$ be the smallest $T$ such that \phenom{} can occur in $G(T)$. Let $\vmu^*$ to be any SPE of $G(T^*)$ where \phenom{} occurs. Denote $\vsigma^* = (\sigma_1^*, \sigma_2^*)$ to be the first round strategy profile in $\vmu^*$. It follows that $\vsigma^*\notin \Nash(G)$, and all strategy profiles in all later rounds in $\vmu^*$ belongs to $\Nash(G)$. Since $|V_2| = 1$, player 2 must play a best response in $\vsigma^*$. Therefore, $\sigma_1^*$ must assign positive probabilities in both $a_1$ and $c_1$, since otherwise $\vsigma^*\in \Nash(G)$. For $\sigma_2^*$, either $\sigma_2^*(a_2)\geq 0.5$ or $\sigma_2^*(b_2)\geq 0.5$. If $\sigma_2^*(b_2)\geq 0.5$, we have $u_1(b_1, \sigma_2^*) - u_1(a_1, \sigma_2^*) \geq 0.5(2 - \alpha)$. Denote $U_1(\vmu^*_{|b_1})$ as the expected total payoff for player 1 in the last $T^*-1$ rounds given player 1 plays $b_1$ in the first round, and similarly for $U_1(\vmu^*_{|a_1})$. For $\vmu^*$ to be an SPE, we must have $U_1(\vmu^*_{|a_1}) - U_1(\vmu^*_{|b_1}) \geq u_1(b_1, \sigma_2^*) - u_1(a_1, \sigma_2^*) \geq 0.5(2-\alpha)$. But we also have $U_1(\vmu^*_{|a_1}) - U_1(\vmu^*_{|b_1}) \leq 3(T^*-1) - 2(T^* - 1) = T^* - 1$, so $T^* > 0.5(2-\alpha)$. The same argument can be applied to the case where $\sigma_2^*(a_2)\geq 0.5$. Therefore, \phenom{} cannot occur with any $T < \frac{1}{2}(2-\alpha)$.

\end{example}

\section{The Pure Strategy Case}
\label{sec:pure}

We start by considering 2-player games where players can only use pure strategies. In this case, a strategy of player $i$ in $G(T)$ is $\mu_i: H\rightarrow A_i$. 
%We use $\Nash^{p,p}(G)$ to denote the set of Nash equilibria of the stage game $G$ when both players can only use pure strategies, $\spe^{p,p}(G,T)$ to denote the set of SPEs of the repeated game $G(T)$ when both players can only use pure strategies, and $V_i^{p,p} = \condSet{u_i(\vsigma)}{\vsigma\in\Nash^{p,p}(G)}$ to denote the set of payoff values attainable at stage game Nash equilibria for player $i$ when both players can only use pure strategies.
We use $\Nash^{p,p}(G)$, $\spe^{p,p}(G,T)$, and $V_i^{p,p}$ to denote the corresponding concepts of $\Nash(G)$, $\spe(G,T)$, and $V_i$ when both players can only use pure strategies.
We use $\glspp$ and $\glopp$ to denote the partition of the set of all stage games $\mathcal{G}$ when both players can only use pure strategies. For any stage game $G$ where $\Nash^{p,p}(G) = \emptyset$, there is no SPE in the repeated game $G(T)$ for any $T$, so $G\in \glopp$ since \phenom{} can never occur. The following theorem presents a complete mathematical characterization of $\glspp$. As we will see, the key requirement for \phenom{} to occur is the ability of some player to `threaten' the other player to play off stage-game NEs in some rounds.

\begin{theorem}[2-player, pure strategy]
\label{thm:pure}
For 2-player pure-strategy-only games, a sufficient and necessary condition on the stage game $G$ for there exists some $T$ and some SPE of $G(T)$ where \phenom{} occurs is:
\begin{enumerate}
    \item $|V_1^{p,p}| > 1$, $|V_2^{p,p}| > 1$, and there exists some $\hat{a}_1\in A_1, \hat{a}_2\in A_2$ where $(\hat{a}_1, \hat{a}_2) \notin \Nash^{p,p}(G)$, OR
    \item $|V_1^{p,p}| > 1$, $|V_2^{p,p}| = 1$, and there exists $\hat{a}_1, a_1' \in A_1, \hat{a}_2 \in A_2$ where $u_1(\hat{a}_1, \hat{a}_2) < u_1(a_1', \hat{a}_2)$ and $\hat{a}_2$ is a best response to $\hat{a}_1$, OR
    \item $|V_1^{p,p}| = 1$, $|V_2^{p,p}| > 1$, and there exists $\hat{a}_1 \in A_1, \hat{a}_2, a_2' \in A_2$ where $u_2(\hat{a}_1, \hat{a}_2) < u_2(\hat{a}_1, a_2')$ and $\hat{a}_1$ is a best response to $\hat{a}_2$.
\end{enumerate}
\end{theorem}

\begin{proof}

First we show the condition is sufficient, by showing if the condition is satisfied, then we can construct some $T$ and some SPE where \phenom{} occurs. 

If the condition is satisfied, then at least one of (1),(2),(3) must be satisfied. If (1) is satisfied, there exists $\va_1, \va_1'\in\Nash^{p,p}(G)$ where  $u_1(\va_1) > u_1(\va'_1)$ and $\va_2, \va_2'\in\Nash^{p,p}(G)$ where $u_2(\va_2) > u_2(\va'_2)$ (note that $\va_1, \va_1'$ do not need to be different from $\va_2, \va_2'$). From the given $(\hat{a}_1, \hat{a}_2) \notin \Nash^{p,p}(G)$, let $\delta_1 = \max_{a_1\in A_1} u_1(a_1, \hat{a}_2) - u_1(\hat{a}_1, \hat{a}_2)$ and $\delta_2 = \max_{a_2\in A_2} u_2(\hat{a}_1, a_2) - u_2(\hat{a}_1, \hat{a}_2)$. We set $T = 2k + 1$ and $k \geq \max \Big( \frac{\delta_1}{u_1(\va_1)-u_1(\va'_1)}, \frac{\delta_2}{u_2(\va_2)-u_2(\va'_2)} \Big)$. We argue that the following strategy profile is an SPE of $G(T)$:
\begin{itemize}
    \item In the first round, play $(\hat{a}_1, \hat{a}_2)$,
    \item For later $2k$ rounds, if the first round play is $(\hat{a}_1, \hat{a}_2)$, players play their corresponding strategy according to $(\va_1, \va_2, \va_1, \va_2, \dots)$; if the first round play is $(a_1', \hat{a}_2)$ where $a_1' \neq \hat{a}_1$, players play their corresponding strategy according to $(\va_1', \va_2, \va_1', \va_2, \dots)$; if the first round play is $(\hat{a}_1, a_2')$ where $a_2' \neq \hat{a}_2$, players play their corresponding strategy according to $(\va_1, \va_2', \va_1, \va_2', \dots)$; otherwise, players play their corresponding strategy according to $(\va_1, \va_1, \dots)$ (or any sequence of NEs).
\end{itemize}
For every subgame in the game tree starting from the second round or later, the above strategy profile forms an NE, since a stage game NE is played in every round. So we only need to show that the above strategy profile forms an NE for the root game $G(T)$. The total payoff for player 1 under the above strategy profile is $U_1 = u_1(\hat{a}_1, \hat{a}_2) + k\cdot(u_1(\va_1) + u_1(\va_2))$. If player 1 unilaterally deviates on the first round play (and possibly later rounds as well), the new total payoff $U_1' \leq \max_{a_1\in A_1} u_1(a_1, \hat{a}_2) + k\cdot(u_1(\va_1') + u_1(\va_2)) \leq U_1$ due to our choice of $k$. The same argument applies for player 2, so the above strategy profile forms an NE for the root game. \Phenom{} occurs here since the first round behavior strategy profile $(\hat{a}_1, \hat{a}_2) \notin \Nash^{p,p}(G)$.

If (2) is satisfied, there exists $\va_1, \va_1'\in\Nash^{p,p}(G)$ where $u_1(\va_1) > u_1(\va'_1)$. Let $\delta = \max_{a_1} u_1(a_1, \hat{a}_2) - u_1(\hat{a}_1, \hat{a}_2)$. We set $T=k+1$ and $k \geq \frac{\delta}{u_1(\va_1) - u_1(\va'_1)}$. We argue that the following strategy profile is an SPE of $G(T)$: 
\begin{itemize}
    \item In the first round, play $(\hat{a}_1, \hat{a}_2)$,
    \item For the later $k$ rounds, if player 1's first round play is $\hat{a}_1$, players play their corresponding strategy according to $(\va_1, \va_1, \dots)$; if player 1's first round play is $a_1'\neq \hat{a}_1$, players play their corresponding strategy according to $(\va_1', \va_1', \dots)$.
\end{itemize}
Again, for every subgame in the game tree starting from the second round or later, the above strategy profile forms an NE. So we only need to show that the above strategy profile forms an NE for the root game. Player 2 cannot deviate to get a higher total payoff since $\hat{a}_2$ is the best response to $\hat{a}_1$ and $u_2$ is the same under all $\Nash^{p,p}(G)$. For player 1, the total payoff under the above strategy profile is $U_1 = u_1(\hat{a}_1, \hat{a}_2) + k\cdot u_1(\va_1)$. If player 1 unilaterally deviates, the new total payoff $U_1' \leq \max_{a_1\in A_1} u_1(a_1, \hat{a}_2) + k\cdot u_1(\va_1') \leq U_1$ due to our choice of $k$. So the above strategy profile forms an NE for the root game. \Phenom{} occurs here since $(\hat{a}_1, \hat{a}_2) \notin \Nash^{p,p}(G)$.

The same construction exchanging player 1 and 2 works for the case when (3) is satisfied. This finishes the proof that the condition is sufficient.
\\
\\
To prove this condition is necessary, we prove that if the condition is not satisfied, then for any $T$ and any SPE $\vmu$ of $G(T)$, \phenom{} does not occur, i.e., the strategy profile at each round must form an NE of the stage game. The condition is not satisfied means all of (1),(2),(3) are false. This can be divided into the following disjoint cases.

\textbf{1. $|V_1^{p,p}| = 0$, $|V_2^{p,p}|=0$.} Here, $\Nash^{p,p}(G) = \emptyset$, so there is no SPE in the repeated game $G(T)$ for any $T$. Therefore, \phenom{} can never occur.

\textbf{2. $|V_1^{p,p}| = 1$, $|V_2^{p,p}|=1$.} 
Using backward induction, we know that in any SPE, the strategy profile at each round must form an NE of the stage game.

\textbf{3. $|V_1^{p,p}| > 1$, $|V_2^{p,p}|>1$.}
Since (1) is false, there does not exist $\hat{a}_1, \hat{a}_2$ where $(\hat{a}_1, \hat{a}_2)\notin \Nash^{p,p}(G)$, i.e., $A = \Nash^{p,p}(G)$ (\Cref{example:every_ne} shows an example of such games). Therefore, it trivially follows that in any SPE of $G(T)$, the strategy profile at each round must form an NE of the stage game.

\textbf{4. $|V_1^{p,p}| > 1$, $|V_2^{p,p}|=1$. }
Since (2) is false, there does not exist $\hat{a}_1, a_1'\in A_1, \hat{a}_2\in A_2$ where $u_1(\hat{a}_1, \hat{a}_2) < u_1(a_1', \hat{a}_2)$ and $\hat{a}_2$ is a best response to $\hat{a}_1$ (\Cref{example:no_best_response} shows an example of such games). This means that for any $\hat{a}_1, \hat{a}_2$ where $\hat{a}_2$ is a best response to $\hat{a}_1$, $\hat{a}_1$ is a best response to $\hat{a}_2$, and thus $(\hat{a}_1, \hat{a}_2) \in \Nash^{p,p}(G)$. Now we can use backward induction to prove that in any SPE, the strategy profile at each round must form an NE of the stage game. The strategy profiles in the last round must form NEs. Given that the strategy profiles in the last $k$ rounds must all form NEs, consider the $(k+1)$-to-last round. Player 2 must play a best response in this round, since their play in this round does not affect the total payoff they get in the final $k$ rounds. And since all strategy profiles where player 2 plays best response is an NE of the stage game, the induction step is complete.

\textbf{5. $|V_1^{p,p}| = 1$, $|V_2^{p,p}|>1$. } The same proof for case 4 applies here.

This finishes the proof that the above condition is necessary.
\end{proof}

A reader may wonder whether there exists stage games $G$ that belong to cases 3 and 4 in the above proof for the necessity of the condition. We present here example games that belong to each case.

\begin{example}
\label{example:every_ne}
\Cref{tab:every_ne} presents an example stage game $G$ where  $|V_1^{p,p}| > 1$, $|V_2^{p,p}|>1$, and there does not exist $\hat{a}_1, \hat{a}_2$ where $(\hat{a}_1, \hat{a}_2)\notin \Nash^{p,p}(G)$, i.e., $A = \Nash^{p,p}(G)$.
\end{example}

\begin{example}
\label{example:no_best_response}
\Cref{tab:no_best_response} presents an example stage game $G$ where $|V_1^{p,p}| > 1$, $|V_2^{p,p}|=1$, and there does not exist $\hat{a}_1, a_1'\in A_1, \hat{a}_2\in A_2$ where $u_1(\hat{a}_1, \hat{a}_2) < u_1(a_1', \hat{a}_2)$ and $\hat{a}_2$ is a best response to $\hat{a}_1$.
The set of pure Nash equilibria of $G$ is: $(a_1, a_2)$, $(b_1, a_2)$, $(b_1, b_2)$, and $(c_1, b_2)$. Therefore, $|V_1^{p,p}| > 1$ and $|V_2^{p,p}|=1$. We can see that for all strategy profiles where player 2 plays a best response, player 1 also plays a best response.

\begin{table}[ht]
    \centering
    \begin{tabular}{c|c|c|}
    & $a_2$ & $b_2$ \\
    \midrule
    $a_1$ & (3,2)   & (1,1) \\
      \midrule
    $b_1$ & (3,2)   & (2,2) \\
      \midrule
    $c_1$ & (1,1)  & (2,2) \\
       \midrule
    \end{tabular}
    \caption{Example stage game in matrix form, row player is player 1, column player is player 2. For all $a\in A_1$, for all $\sigma_2\in \Delta A_2$ that is a best response to $a$, $a$ is also a best response to $\sigma_2$. }
    \label{tab:no_best_response}
\end{table}
\end{example}

From the constructions of SPEs where \phenom{} occurs used in the above proof, we can obtain the following corollary regarding the value of $T$ above which \phenom{} can occur if the condition in \Cref{thm:pure} is satisfied:

\begin{corollary}
For 2-player pure-strategy-only games, given a stage game $G$:
\begin{enumerate}
    \item If $|V_1^{p,p}| > 1$, $|V_2^{p,p}| > 1$, and there exists some $\hat{a}_1\in A_1, \hat{a}_2\in A_2$ where $(\hat{a}_1, \hat{a}_2) \notin \textrm{Nash}(G)$, then for all $T \geq 2\cdot \max\Big( \frac{\delta_1}{\max (V_1^{p,p})-\min (V_1^{p,p})}, \frac{\delta_2}{\max (V_2^{p,p})- \min (V_2^{p,p})} \Big) +1$ where $\delta_1 = \max_{a_1\in A_1} u_1(a_1, \hat{a}_2) - u_1(\hat{a}_1, \hat{a}_2)$ and $\delta_2 = \max_{a_2\in A_2} u_2(\hat{a}_1, a_2) - u_2(\hat{a}_1, \hat{a}_2)$, there exists some SPE of $G(T)$ where \phenom{} occurs.
    \item If $|V_1^{p,p}| > 1$, $|V_2^{p,p}| = 1$, and there exists $\hat{a}_1, a_1' \in A_1, \hat{a}_2 \in A_2$ where $u_1(\hat{a}_1, \hat{a}_2) < u_1(a_1', \hat{a}_2)$ and $\hat{a}_2$ is a best response to $\hat{a}_1$, then for all $T\geq \frac{\max_{a_1\in A_1} u_1(a_1, \hat{a}_2) - u_1(\hat{a}_1, \hat{a}_2)}{\max (V_1^{p,p})-\min (V_1^{p,p})} + 1$, there exists some SPE of $G(T)$ where \phenom{} occurs.
    \item If $|V_1^{p,p}| = 1$, $|V_2^{p,p}| > 1$, and there exists $\hat{a}_1 \in A_1, \hat{a}_2, a_2' \in A_2$ where $u_2(\hat{a}_1, \hat{a}_2) < u_2(\hat{a}_1, a_2')$ and $\hat{a}_1$ is a best response to $\hat{a}_2$, then for all $T\geq \frac{\max_{a_2\in A_2} u_2(\hat{a}_1, a_2) - u_2(\hat{a}_1, \hat{a}_2)}{\max (V_2^{p,p})-\min (V_2^{p,p})} + 1$, there exists some SPE of $G(T)$ where \phenom{} occurs.
\end{enumerate}
\end{corollary}

\section{The General Case}
\label{sec:mixed}

Now we consider the general case for 2-player games where mixed strategies are allowed. 
%We use $\Nash^{m,m}(G)$ to denote the set of Nash equilibria of the stage game $G$ when both players can use mixed strategies, and $V_i^{m,m} = \condSet{u_i(\vsigma)}{\vsigma\in\Nash^{m,m}(G)}$ to denote the set of payoff values attainable at stage game Nash equilibria for player $i$ when both players can use mixed strategies. 
We use $\Nash^{m,m}(G)$, $\spe^{m,m}(G,T)$, and $V_i^{m,m}$ to denote the corresponding concepts of $\Nash(G)$, $\spe(G,T)$, and $V_i$ when both players can use mixed strategies.
Since mixed Nash equilibrium always exists for $G$ \cite{nash1951non}, $|V_1^{m,m}|\geq 1$, $|V_2^{m,m}|\geq 1$. We use $\glsmm$ and $\glomm$ to denote the partition of the set of all stage games $\mathcal{G}$ when both players can use mixed strategies. The following theorem presents a complete mathematical characterization of $\glsmm$.

\begin{theorem} [2-player, general case]
\label{thm:general}
For general 2-player games (mixed strategies allowed),
a sufficient and necessary condition on the stage game $G$ for there exists some $T$ and some SPE of $G(T)$ where \phenom{} occurs is:
\begin{enumerate}
    \item $|V_1^{m,m}| > 1$, $|V_2^{m,m}| > 1$, and there exists some $\hat{\sigma}_1 \in \Delta A_1, \hat{\sigma}_2\in \Delta A_2$ where $(\hat{\sigma}_1, \hat{\sigma}_2) \notin \Nash^{m,m}(G)$, OR
    \item $|V_1^{m,m}| > 1$, $|V_2^{m,m}| = 1$, and there exists $\hat{\sigma}_1\in \Delta A_1, \hat{\sigma}_2 \in \Delta A_2, a_1'\in A_1$ where $u_1(\hat{\sigma}_1, \hat{\sigma}_2) < u_1(a_1', \hat{\sigma}_2)$ and $\hat{\sigma}_2$ is a best response to $\hat{\sigma}_1$, OR
    \item same as (2) but exchange player 1 and 2.
\end{enumerate}
\end{theorem}

We first establish some useful lemmas.

\begin{lemma}
\label{lemma:mixed-pure}
For any two-player game $G$, if there exists $\sigma_1\in\Delta A_1$ and $\sigma_2\in\Delta A_2$ where $(\sigma_1, \sigma_2)\notin \Nash^{m,m}(G)$, then there exists $a_1\in A_1$ and $a_2\in A_2$ where $(a_1, a_2) \notin \Nash^{m,m}(G)$.
\end{lemma}

\begin{proof}
$(\sigma_1, \sigma_2)\notin \Nash^{m,m}(G)$ implies that there exists some $a_1'\in A_1$ where $u_1(\sigma_1, \sigma_2) < u_1(a_1', \sigma_2)$, or there exists some $a_2'\in A_2$ where $u_2(\sigma_1, \sigma_2) < u_2(\sigma_1, a_2')$. We consider the case of there exists some $a_1'$ where $u_1(\sigma_1, \sigma_2) < u_1(a_1', \sigma_2)$, and the same argument applies to the other case. Since $u_1(\sigma_1, \sigma_2) \geq \min_{a_1\in S_{\sigma_1}} u_1 (a_1, \sigma_2)$, there exists some $a_1, a_1', \sigma_2$ where $u_1(a_1, \sigma_2) < u_1(a_1', \sigma_2)$. So $u_1(a_1', \sigma_2) - u_1(a_1, \sigma_2) = \sum_{a_2\in S_{\sigma_2}} \sigma_2(a_2)\cdot \Big(u_1(a_1', a_2) - u_1(a_1, a_2) \Big) > 0$, which means there exists some $a_2$ where $u_1(a_1', a_2) - u_1(a_1, a_2) > 0$. This $(a_1, a_2)\notin \Nash^{m,m}(G)$, which finishes the proof.
\end{proof}

\begin{lemma}
\label{lemma:matrix}
For any two-player game $G$, define 
\begin{align*}
    I = \condSet{(i,j)}{i\in A_1, j\in A_2, u_2(i, j) = \max_{j'\in A_2} u_2(i, j')}
\end{align*}
as the set of pure strategy profiles where player 2 plays a best response. If
\begin{enumerate}
    \item $|V_1^{m,m}| > 1$, $|V_2^{m,m}|=1$, and
    \item there does not exist $\hat{a}_1\in A_1, \hat{\sigma}_2 \in \Delta A_2, a_1'\in A_1$ where $u_1(\hat{a}_1, \hat{\sigma}_2) < u_1(a_1', \hat{\sigma}_2)$ and $\hat{\sigma}_2$ is a best response to $\hat{a}_1$, and
    \item there does not exist $a_1\in A_1$ and $a_2, a_2' \in A_2$ where $a_2\neq a_2'$ and both $a_2$ and $a_2'$ are best responses to $a_1$,
\end{enumerate}
then, 
\begin{itemize}
    \item [(a).] $I\subseteq \Nash^{m,m}(G)$,
    \item [(b).] for each $i\in A_1$, there is a unique $j\in A_2$ such that $(i,j)\in I$,
    \item [(c).] there exists $b\in \real$ such that for all $(i,j)\in I$, $u_2(i, j) = b$, and for all $(i',j')\notin I$, $u_2(i', j') < b$,
    \item [(d).] there exists $(i,j), (i',j')\in I$ such that $u_1(i, j) \neq u_1(i',j')$, and $i\neq i'$, $j\neq j'$.
\end{itemize}
\end{lemma}

\begin{proof}
(2) directly implies (a). From the definition of $I$, for each $i\in A_1$, there is at least one $j\in A_2$ such that $(i,j)\in I$. This combines with (3) implies (b).

Since $|V_2^{m,m}|=1$ and $I\subseteq \Nash^{m,m}(G)$, for all $(i,j)\in I$, $u_2(i,j)=b$ where $b$ is the only element in $V_2^{m,m}$. It then follows from the definition of $I$ that for all $(i',j')\notin I$, $u_2(i', j') < b$. So (c) follows.

For (d), assume in contradiction that all $u_1(i,j)$ for $(i,j)\in I$ are the same. Since $|V_1^{m,m}|>1$, there exists $\vsigma, \vsigma'\in \Nash^{m,m}(G)$ such that $u_1(\vsigma) \neq u_1(\vsigma')$. Denote $\mathcal{S}_{\vsigma}$ as the set of pure strategy profiles $(i,j)$ that occur with non-zero probability under the strategy profile $\vsigma$. Then at least one of $\mathcal{S}_{\vsigma}$ and $\mathcal{S}_{\vsigma'}$ needs to contain elements not in $I$, since otherwise $u_1(\vsigma) = u_1(\vsigma')$. WLOG, let $\mathcal{S}_{\vsigma}$ contain elements not in $I$. By (c), $u_2(\vsigma) < b$, which contradicts with $|V_2^{m,m}|=1$. So there exists $(i,j), (i',j')\in I$ such that $u_1(i, j) \neq u_1(i',j')$. For such $(i,j), (i',j')$, if $i=i'$, then (b) implies that $j=j'$, which contradicts with $u_1(i, j) \neq u_1(i',j')$. So $i\neq i'$. And since $I\subseteq \Nash^{m,m}(G)$ (due to (a)), $(i,j), (i',j')$ are both NEs with different payoffs for player 1, so $j\neq j'$. Therefore, (d) follows. 

\end{proof}

Now we are ready to prove \Cref{thm:general}.

\begin{proof}[Proof of \Cref{thm:general}]

Following the same argument as the proof for the pure strategy case (\Cref{thm:pure}), we can show the condition is necessary.

We prove the condition is sufficient by showing if the condition is satisfied, we can always construct some $T$ and some SPE where \phenom{} occurs. If the condition is satisfied, then at least one of (1),(2),(3) must be satisfied. We consider each case here.

If (1) is satisfied, there exists some $\hat{\sigma}_1 \in \Delta A_1, \hat{\sigma}_2\in \Delta A_2$ where $(\hat{\sigma}_1, \hat{\sigma}_2) \notin \Nash^{m,m}(G)$. By \Cref{lemma:mixed-pure}, there exists $\hat{a}_1\in A_1$ and $\hat{a}_2\in A_2$ where $(\hat{a}_1, \hat{a}_2) \notin \Nash^{m,m}(G)$. Then we can use the same construction that is used in the proof of the pure strategy case here (see the proof of sufficiency in \Cref{thm:pure}, the part that handles the case where (1) is satisfied).

The rest of the proof focus on the case when (2) is satisfied. The same argument applies for the case where (3) is satisfied. We first notice that all games that satisfy (2) can be categorized into the following 3 disjoint cases: 
\begin{itemize}
    \item[(a).] There exists $\hat{a}_1\in A_1, \hat{\sigma}_2 \in \Delta A_2, a_1'\in A_1$ where $u_1(\hat{a}_1, \hat{\sigma}_2) < u_1(a_1', \hat{\sigma}_2)$ and $\hat{\sigma}_2$ is a best response to $\hat{a}_1$.
    \item[(b).] (a) is false, and there exists $a_1\in A_1$ and $a_2, a_2' \in A_2$ where $a_2\neq a_2'$ and both $a_2$ and $a_2'$ are best responses to $a_1$.
    \item[(c).] Both (a) and (b) are false.
\end{itemize}
We consider each case here.
\\
\\
\textbf{Case (a).} We can use the same construction that is used in the proof of sufficiency for the pure strategy case (\Cref{thm:pure}), the part that handles the case where (2) is satisfied.
\\
\\
\textbf{Case (b).} (a) is false implies that for all $a_1\in A_1$, for all $\sigma_2\in \Delta A_2$ that is a best response to $a_1$, $a_1$ is also a best response to $\sigma_2$. \Cref{tab:no_best_response} is an example of such games. We know that there exists some $a^*\in A_1$, $b_1^*, b_2^* \in A_2$ where $b_1^*\neq b_2^*$ and both $b_1^*$ and $b_2^*$ are best responses to $a^*$. Therefore, $a^*$ is a best response to both $b_1^*$ and $b_2^*$. WLOG, let $u_1(a^*, b_1^*) \geq u_1(a^*, b_2^*)$. Denote $\sigma_{\lambda}\in \Delta A_2$ as the mixed strategy for player 2 which assigns $\sigma_{\lambda}(b_1^*) = \lambda$ and $\sigma_{\lambda}(b_2^*) = 1 - \lambda$. Then for all $0\leq \lambda \leq 1$, $(a^*, \sigma_{\lambda})\in \Nash^{m,m}(G)$.

Since (2) is satisfied, there exists $\hat{\sigma}_1\in \Delta A_1, \hat{\sigma}_2 \in \Delta A_2, a_1'\in A_1$ where $u_1(\hat{\sigma}_1, \hat{\sigma}_2) < u_1(a_1', \hat{\sigma}_2)$ and $\hat{\sigma}_2$ is a best response to $\hat{\sigma}_1$. We construct $T = 1 + T_1 + T_2$ and a strategy profile $\vmu^*$ for $G(T)$ with the following structure:
\begin{itemize}
    \item In the first round, play $(\hat{\sigma}_1, \hat{\sigma}_2)$.
    \item For the later rounds, if player 1's first round play is $i\in S_{\hat{\sigma}_1}$, players play their corresponding strategies according to SPE $\vmu^i$ of $G(T-1)$; otherwise, players play their corresponding strategies according to SPE $\vmu^{\bot}$ of $G(T-1)$.
\end{itemize}

Since $|V_1^{m,m}| > 1$, let $\vsmin, \vsmax\in \Nash^{m,m}(G)$ such that $u_1(\vsmin) = \min (V_1^{m,m})$ and $u_1(\vsmax) = \max (V_1^{m,m})$, so $u_1(\vsmax) > u_1(\vsmin)$. We construct the SPEs $\vmu^{\bot}, \{\vmu^i\}_{i\in S_{\hat{\sigma}_1}}$ as follows:

\begin{itemize}
    \item $\vmu^{\bot}$ is players playing $\vsmin$ repeatedly for $T_1 + T_2$ rounds.
    \item For all $\vmu^i$, the last $T_2$ rounds consist of players repeated playing $\vsmax$. $T_2$ is chosen to be large enough such that $U_1(\vmu^i) - U_1(\vmu^{\bot}) > \max_{a,a'\in A_1} u_1(a,\hat{\sigma}_2) - u_1(a',\hat{\sigma}_2)$ for every $i\in S_{\hat{\sigma}_1}$. This makes sure that in $\vmu^*$, player 1 deviating to any $i\notin S_{\hat{\sigma}_1}$ in the first round will reduce their total payoff in $G(T)$.
    \item The first $T_1$ rounds strategies for each $\vmu^i$ adopt the following structure:
    \begin{itemize}
        \item In the first round, play $(a^*, \sigma_{\lambda_i})$, where $\lambda_i$ is a parameter to be set for each $i$.
        \item In the latter $T_1 - 1$ rounds, if player 2 plays $b_1^*$ in the first round, players repeatedly play $\vsmax$; otherwise, players repeatedly play $\vsmin$. 
    \end{itemize}
    Pick $i^m \in S_{\hat{\sigma}_1}$ such that $u_1(i^m, \hat{\sigma}_2) = \max_{i \in S_{\hat{\sigma}_1}} u_1(i, \hat{\sigma}_2)$. We set $\lambda_{i^m} = 0$. For each $i\in S_{\hat{\sigma}_1} \setminus \{i^m\}$, we set $\lambda_i$ such that $U_1(\vmu^i) + u_1(i,\hat{\sigma}_2) = U_1(\vmu^{i^m}) + u_1(i^m,\hat{\sigma}_2)$. This makes sure that in $\vmu^*$, player 1 choosing any $i\in S_{\hat{\sigma}_1}$ in the first round will obtain the same total payoff in $G(T)$. We argue that with large enough $T_1$, such choice of $\lambda_i$'s is always possible. Consider the difference between two sides of the equation as a function of $\lambda_i$, $f(\lambda_i) = U_1(\vmu^i) - U_1(\vmu^{i^m}) + u_1(i,\hat{\sigma}_2) - u_1(i^m,\hat{\sigma}_2)$. By choosing $T_1 \geq \frac{\max_{i\in S_{\hat{\sigma}_1}} u_1(i^m, \hat{\sigma}_2) - u_1(i, \hat{\sigma}_2) }{\max(V_1^{m,m}) - \min(V_1^{m,m})} + 1$, we have $f(0)\leq 0$ and $f(1) \geq 0$. Since $f(\lambda_i)$ is a continuous function, there must exist some $\lambda_i \in [0,1]$ such that $f(\lambda_i) = 0$ as desired.
\end{itemize}

One can easily verify that $\vmu^{\bot}$ and $\{\vmu^i\}_{i\in S_{\hat{\sigma}_1}}$ are SPEs of $G(T-1)$. In addition, their construction ensures that $\vmu^*$ is an NE of the root game $G(T)$. Therefore, $\vmu^*$ is an SPE of $G(T)$ where the first round strategy profile does not form an NE of the stage game $G$.
\\
\\
\textbf{Case (c).} Since both (a) and (b) are false, and $|V_1^{m,m}| > 1$, $|V_2^{m,m}| = 1$, applying \Cref{lemma:matrix}, we know that there exists $(i_1,j_1), (i_2,j_2)\in I$ such that $i_1\neq i_2$, $j_1\neq j_2$, where $I$ is the set of pure strategy profiles where player 2 plays a best response, as defined in \Cref{lemma:matrix}. Take such $(i_1,j_1), (i_2,j_2)$, \Cref{lemma:matrix} further implies that for all $j\neq j_1$, $u_2(i_1, j) < b$, and for all $j\neq j_2$, $u_2(i_2, j) < b$, where $b$ is the only element in $V_2^{m,m}$. Denote $\hat{\sigma}_{\lambda}\in \Delta A_1$ as the mixed strategy for player 1 which assigns $\hat{\sigma}_{\lambda}(i_1) = \lambda$ and $\hat{\sigma}_{\lambda}(i_2) = 1 - \lambda$. Denote $J(\lambda) = \condSet{a_2}{a_2\in A_2, a_2 \textrm{ is a best response to }\hat{\sigma}_{\lambda}}$ as the set of best response pure strategies for player 2 against $\hat{\sigma}_{\lambda}$. It is helpful to consider a geometric interpretation of $J(\lambda)$. For each $j\in A_2$, $u_2(\hat{\sigma}_{\lambda}, j) = \lambda\cdot u_2(i_1, j) + (1 - \lambda)\cdot u_2(i_2, j)$ is a linear function in $\lambda$. We can plot the function $f_j(\lambda)=u_2(\hat{\sigma}_{\lambda}, j)$ for each $j\in A_2$, which gives $|A_2|$ straight lines within domain $[0,1]$. $J(\lambda)$ is then the set of lines that attains the maximum value at $\lambda$. We know that $J(0) = \{j_2\}$ and $J(1) = \{j_1\}$, so there must exist some $\lambda_1\in (0,1)$ where $|J(\lambda_1)| > 1$, which corresponds to some intersection point. Take such $\lambda_1$ and $\hat{j}_1, \hat{j}_2 \in J(\lambda_1)$ where $\hat{j}_1\neq \hat{j}_2$. Denote $\hat{\sigma}_{\rho}\in \Delta A_2$ as the mixed strategy for player 2 which assigns $\hat{\sigma}_{\rho}(\hat{j}_2) = \rho$ and $\hat{\sigma}_{\rho}(\hat{j}_1) = 1 - \rho$. Then for all $\rho\in[0,1]$, $\hat{\sigma}_{\rho}$ is a best response to $\hat{\sigma}_{\lambda_1}$, which implies that $\hat{\sigma}_{\lambda_1}$ is not a best response to $\hat{\sigma}_{\rho}$. This is because if $\hat{\sigma}_{\lambda_1}$ is a best response to $\hat{\sigma}_{\rho}$, then $(\hat{\sigma}_{\lambda_1}, \hat{\sigma}_{\rho})\in \Nash^{m,m}(G)$, but $u_2(\hat{\sigma}_{\lambda_1}, \hat{\sigma}_{\rho}) < b$, which contradicts with $|V_2^{m,m}| = 1$.

Now we show that we can always construct some $T$ and some SPE $\vmu^*$ of $G(T)$ where the first round strategy profile is $(\hat{\sigma}_{\lambda_1}, \hat{\sigma}_{\rho})$ for some $\rho$. Since the above argument shows that $(\hat{\sigma}_{\lambda_1}, \hat{\sigma}_{\rho})\notin \Nash^{m,m}(G)$, \phenom{} occurs in $\vmu^*$. We treat the case where $u_1(i_1, \hat{j}_1) = u_1(i_2, \hat{j}_1)$ and $u_1(i_1, \hat{j}_1) \neq u_1(i_2, \hat{j}_1)$ separately.

If $u_1(i_1, \hat{j}_1) = u_1(i_2, \hat{j}_1)$, we construct $\vmu^*$ as: 
\begin{itemize}
    \item In the first round, play $(\hat{\sigma}_{\lambda}, \hat{j}_1)$. ($\hat{j}_1$ is $\hat{\sigma}_{\rho}$ with $\rho = 0$)
    \item For the later rounds, if player 1's first round play is $i_1$ or $i_2$, players repeatedly play $\vsmax$; otherwise, players repeatedly play $\vsmin$.
\end{itemize}
Again, $\vsmin, \vsmax\in \Nash^{m,m}(G)$ such that $u_1(\vsmin) = \min (V_1^{m,m})$ and $u_1(\vsmax) = \max (V_1^{m,m})$. $T$ is chosen to be large enough such that player 1 deviating to any $i\notin \{i_1, i_2\}$ in the first round will reduce their total payoff in $G(T)$. It can be easily checked that $\vmu^*$ is an SPE of $G(T)$.

If $u_1(i_1, \hat{j}_1) \neq u_1(i_2, \hat{j}_1)$, WLOG, assume $u_1(i_1, \hat{j}_1) > u_1(i_2, \hat{j}_1)$. We construct $T = 1 + T_1 + T_2$ and $\vmu^*$ with the following structure:
\begin{itemize}
    \item In the first round, play $(\hat{\sigma}_{\lambda}, \hat{\sigma}_{\rho})$.
    \item For the later rounds, if the first round play is $(i_1, \hat{j}_2)$, players play their corresponding strategy according to SPE $\vmu^1$ of $G(T-1)$; if the first round play is $(i_2, \hat{j}_2)$, $(i_1, \hat{j}_1)$ or $(i_2, \hat{j}_1)$, players play their corresponding strategy according to SPE $\vmu^2$ of $G(T-1)$; otherwise, players play their corresponding strategy according to SPE $\vmu^{\bot}$ of $G(T-1)$.
    \begin{itemize}
        \item $\vmu^{\bot}$ is players play $\vsmin$ repeatedly for $T_1 + T_2$ rounds.
        \item $\vmu^2$ is players play $\vsmax$ repeatedly for $T_1 + T_2$ rounds.
        \item $\vmu^1$ is players play $\vsmin$ repeatedly for $T_1$ rounds, and then $\vsmax$ for $T_2$ rounds.
    \end{itemize}
\end{itemize}
$T_2$ is chosen to be large enough such that player 1 deviating to any $i\notin \{i_1, i_2\}$ in the first round will reduce their total payoff in $G(T)$. In order for $\vmu^*$ to be an NE of the root game, player 1 choosing $i_1$ and $i_2$ in the first round need to yield the same total payoff in $G(T)$. The total payoff of player 1 achieved by choosing $i_1$ in the first round under $\vmu^*$ is $\rho\cdot \Big(u_1(i_1, \hat{j}_2) + U_1(\vmu^1) \Big) + (1-\rho)\cdot \Big( u_1(i_1, \hat{j}_1) + U_1(\vmu^2) \Big)$, and the total payoff by choosing $i_2$ is $\rho\cdot \Big(u_1(i_2, \hat{j}_2) + U_1(\vmu^2) \Big) + (1-\rho)\cdot \Big( u_1(i_2, \hat{j}_1) + U_1(\vmu^2) \Big)$. Consider the difference between these two quantities as a function of $\rho$, $g(\rho) = \rho \cdot \Big(U_1(\vmu^1) - U_1(\vmu^2)\Big) + \rho \cdot \Big(u_1(i_1, \hat{j}_2) - u_1(i_2, \hat{j}_2)\Big) + (1-\rho) \cdot \Big( u_1(i_1, \hat{j}_1) - u_1(i_2, \hat{j}_1) \Big)$. We have $g(0) > 0$. By choosing a large enough $T_1$, $g(1) < 0$. Since $g(\rho)$ is a continuous function, there must exist some $\rho \in (0,1)$ such that $g(\rho) = 0$. With this value of $\rho$, $\vmu^*$ is an SPE of $G(T)$ as desired.

\end{proof}

Again, from the constructions of SPEs where \phenom{} occurs used in the above proof, we can obtain the following corollary regarding the value of $T$ above which \phenom{} can occur if the condition in \Cref{thm:general} is satisfied:

\begin{corollary}
For general 2-player games (mixed strategies allowed), given a stage game $G$:
\begin{enumerate}
    \item If $|V_1^{m,m}| > 1$, $|V_2^{m,m}| > 1$, and there exists some $\hat{\sigma}_1 \in \Delta A_1, \hat{\sigma}_2\in \Delta A_2$ where $(\hat{\sigma}_1, \hat{\sigma}_2) \notin \Nash^{m,m}(G)$, then by \Cref{lemma:mixed-pure}, there exists $\hat{a}_1\in A_1$ and $\hat{a}_2\in A_2$ where $(\hat{a}_1, \hat{a}_2) \notin \Nash^{m,m}(G)$, for all $T \geq 2\cdot \max\Big( \frac{\delta_1}{\max (V_1^{m,m})-\min (V_1^{m,m})}, \frac{\delta_2}{\max (V_2^{m,m})- \min (V_2^{m,m})} \Big) +1$ where $\delta_1 = \max_{a_1\in A_1} u_1(a_1, \hat{a}_2) - u_1(\hat{a}_1, \hat{a}_2)$ and $\delta_2 = \max_{a_2\in A_2} u_2(\hat{a}_1, a_2) - u_2(\hat{a}_1, \hat{a}_2)$, there exists some SPE of $G(T)$ where \phenom{} occurs.
    \item If $|V_1^{m,m}| > 1$, $|V_2^{m,m}| = 1$, and there exists $\hat{\sigma}_1\in \Delta A_1, \hat{\sigma}_2 \in \Delta A_2, a_1'\in A_1$ where $u_1(\hat{\sigma}_1, \hat{\sigma}_2) < u_1(a_1', \hat{\sigma}_2)$ and $\hat{\sigma}_2$ is a best response to $\hat{\sigma}_1$, then
    \begin{enumerate}
        \item If there exists $\hat{a}_1\in A_1, \hat{\sigma}_2 \in \Delta A_2, a_1'\in A_1$ where $u_1(\hat{a}_1, \hat{\sigma}_2) < u_1(a_1', \hat{\sigma}_2)$ and $\hat{\sigma}_2$ is a best response to $\hat{a}_1$, then for all $T\geq \frac{\max_{a_1\in A_1} u_1(a_1, \hat{\sigma}_2) - u_1(\hat{a}_1, \hat{\sigma}_2)}{\max (V_1^{m,m})-\min (V_1^{m,m})} + 1$, there exists some SPE of $G(T)$ where \phenom{} occurs.
        \item If there does not exist $\hat{a}_1\in A_1, \hat{\sigma}_2 \in \Delta A_2, a_1'\in A_1$ where $u_1(\hat{a}_1, \hat{\sigma}_2) < u_1(a_1', \hat{\sigma}_2)$ and $\hat{\sigma}_2$ is a best response to $\hat{a}_1$, and there exists $a_1\in A_1$ and $a_2, a_2' \in A_2$ where $a_2\neq a_2'$ and both $a_2$ and $a_2'$ are best responses to $a_1$, then for all $T \geq 3 + \frac{\max_{a, a'\in S_{\hat{\sigma}_1}} u_1(a, \hat{\sigma}_2) - u_1(a', \hat{\sigma}_2) }{\max(V_1^{m,m}) - \min(V_1^{m,m})} + \frac{\max_{a,a'\in A_1} u_1(a,\hat{\sigma}_2) - u_1(a',\hat{\sigma}_2)}{\max(V_1^{m,m}) - \min(V_1^{m,m})}$, there exists some SPE of $G(T)$ where \phenom{} occurs.
        \item If there does not exist $\hat{a}_1\in A_1, \hat{\sigma}_2 \in \Delta A_2, a_1'\in A_1$ where $u_1(\hat{a}_1, \hat{\sigma}_2) < u_1(a_1', \hat{\sigma}_2)$ and $\hat{\sigma}_2$ is a best response to $\hat{a}_1$, and there does not exist $a_1\in A_1$ and $a_2, a_2' \in A_2$ where $a_2\neq a_2'$ and both $a_2$ and $a_2'$ are best responses to $a_1$. By \Cref{lemma:matrix}, there exists $(i_1,j_1), (i_2,j_2)\in I$ such that $i_1\neq i_2$, $j_1\neq j_2$, where $I$ is the set of pure strategy profiles where player 2 plays a best response. Denote $\hat{\sigma}_{\lambda}\in \Delta A_1$ as the mixed strategy for player 1 which assigns $\hat{\sigma}_{\lambda}(i_1) = \lambda$ and $\hat{\sigma}_{\lambda}(i_2) = 1 - \lambda$. By the proof of \Cref{thm:general}, there exists $\lambda_1\in (0,1)$, $\hat{j}_1, \hat{j}_2 \in A_2$ where $\hat{j}_1, \hat{j}_2$ are both best responses to $\hat{\sigma}_{\lambda_1}$ and $\hat{j}_1\neq \hat{j}_2$. If $u_1(i_1, \hat{j}_1) = u_1(i_2, \hat{j}_1)$, then for all $T\geq \frac{\max_{a_1\in A_1} u_1(a_1, \hat{j}_1) - u_1(i_1, \hat{j}_1)}{\max (V_1^{m,m})-\min (V_1^{m,m})} + 1$, there exists some SPE of $G(T)$ where \phenom{} occurs. If $u_1(i_1, \hat{j}_1) \neq u_1(i_2, \hat{j}_1)$, then for all $T \geq 3 + \frac{\max_{a_1\in A_1, i\in \{i_1, i_2\}, j\in \{ \hat{j}_1, \hat{j}_2 \}} u_1(a_1, j) - u_1(i, j) }{\max (V_1^{m,m})-\min (V_1^{m,m})} + \frac{|u_1(i_1, \hat{j}_2) - u_1(i_2, \hat{j}_2)|}{\max (V_1^{m,m})-\min (V_1^{m,m})}$, there exists some SPE of $G(T)$ where \phenom{} occurs.
    \end{enumerate}
    \item same as (2) but exchange player 1 and 2.
\end{enumerate}
\end{corollary}

\section{The Pure Strategy Against Mixed Strategy Case}
\label{sec:pure_vs_mixed}

To complete the picture, we also analyze the case where one player can only use pure strategies while the other player can use mixed strategies. Without loss of generality, we consider the case where player 1 can use mixed strategies and player 2 can only use pure strategies in both the stage game and the repeated games. 
%We use $\Nash^{m,p}(G)$ to denote the set of Nash equilibria of the stage game $G$ when player 1 can use mixed strategies and player 2 can only use pure strategies, and $V_i^{m,p} = \condSet{u_i(\vsigma)}{\vsigma\in\Nash^{m,p}(G)}$ to denote the set of payoff values attainable at stage game Nash equilibria for player $i$ when player 1 can use mixed strategies and player 2 can only use pure strategies. 
We use $\Nash^{m,p}(G)$, $\spe^{m,p}(G,T)$, and $V_i^{m,p}$ to denote the corresponding concepts of $\Nash(G)$, $\spe(G,T)$, and $V_i$ when player 1 can use mixed strategies and player 2 can only use pure strategies.
We use $\glsmp$ and $\glomp$ to denote the partition of the set of all stage games $\mathcal{G}$ when player 1 can use mixed strategies and player 2 can only use pure strategies. For stage games $G$ where $\Nash^{m,p}(G) = \emptyset$, there is no SPE in the repeated game $G(T)$ for any $T$, so we categorize such stage games $G$ to $\glomp$ since \phenom{} can never occur. The following theorem presents a complete mathematical characterization of $\glsmp$. As we will see, the sufficient and necessary condition for \phenom{} to occur in this case is different from both the pure strategy case and the general case.

\begin{theorem}[2-player, pure strategy against mixed strategy]
\label{thm:pure-vs-mixed}
For 2-player games where player 1 can use mixed strategies and player 2 can only use pure strategies, a sufficient and necessary condition on the stage game $G$ for there exists some $T$ and some SPE of $G(T)$ where \phenom{} occurs is:
\begin{enumerate}
    \item $|V_1^{m,p}| > 1$, $|V_2^{m,p}| > 1$, and there exists some $\hat{\sigma}_1\in \Delta A_1, \hat{a}_2\in A_2$ where $(\hat{\sigma}_1, \hat{a}_2) \notin \Nash^{m,p}(G)$, OR
    \item $|V_1^{m,p}| > 1$, $|V_2^{m,p}| = 1$, and there exists $\hat{\sigma}_1 \in \Delta A_1, a_1' \in A_1, \hat{a}_2 \in A_2$ where
    \begin{enumerate}
        \item $u_1(\hat{\sigma}_1, \hat{a}_2) < u_1(a_1', \hat{a}_2)$, and
        \item $\hat{a}_2$ is a best response to $\hat{\sigma}_1$, and
        \item if $\hat{\sigma}_1$ has more than one support (not a pure strategy), denote the set of possible differences in $u_1$ between pairs of NEs in the stage game as $D = \condSet{u_1(\vsigma) - u_1(\vsigma')}{\vsigma,\vsigma' \in \Nash^{m,p}(G)}$, there exists an action $a$ from the support of $\hat{\sigma}_1$, i.e., $a\in S_{\hat{\sigma}_1}$, such that, for every $a'\in S_{\hat{\sigma}_1}\setminus a$, there exists some integer $n_{a'}\geq 0$ and $d_k^{a'}\in D$, $k=1,\dots,n_{a'}$ such that $u_1(a, \hat{a}_2) - u_1(a', \hat{a}_2) = \sum_{k=1}^{n_{a'}} d_k^{a'}$,
    \end{enumerate} OR
    \item $|V_1^{m,p}| = 1$, $|V_2^{m,p}| > 1$, and there exists $\hat{\sigma}_1 \in \Delta A_1, \hat{a}_2, a_2' \in A_2$ where $u_2(\hat{\sigma}_1, \hat{a}_2) < u_2(\hat{\sigma}_1, a_2')$ and $\hat{\sigma}_1$ is a best response to $\hat{a}_2$.
\end{enumerate}
\end{theorem}

We first establish a useful lemma.

\begin{lemma}
\label{lemma:spe_payoff}
For 2-player stage game $G$ and $T$-round repeated game $G(T)$ where player 1 can use mixed strategies and player 2 can only use pure strategies, for any SPE $\vmu$ of $G(T)$ where the strategy profile at each round forms a stage-game NE, player 1's total payoff in the repeated game $U_1(\vmu) = \sum_{k=1}^{T} c_k$ for some $c_k\in V_1^{m,p}$, $k=1,\dots,T$, i.e., player 1's total payoff in the repeated game equals the sum of some sequence of stage-game NE payoffs.
\end{lemma}

\begin{proof}
We prove by induction on $T$. The proposition trivially holds for $T=1$. Given that the proposition holds for $T=K-1$, consider $T=K$. For any SPE $\vmu$ of $G(K)$ where the strategy profile at each round forms a stage-game NE, denote the first round strategy profile as $(\hat{\sigma}_1, \hat{a}_2)$ where $\hat{\sigma}_1\in \Delta A_1$, $\hat{a}_2\in A_2$. Since $(\hat{\sigma}_1, \hat{a}_2)\in \Nash^{m,p}(G)$, for all $a, a'\in S_{\hat{\sigma}_1}$, $u_1(a, \hat{a}_2) = u_1(a', \hat{a}_2) = u_1(\hat{\sigma}_1, \hat{a}_2)$. Denote $\vmu_{|(a_1, a_2)}$ as the strategy profile starting from round 2 given that players play $(a_1, a_2)$ in the first round, for $\vmu$ to be a Nash equilibrium of the root game, we have for all $a, a'\in S_{\hat{\sigma}_1}$, $u_1(a, \hat{a}_2) + U_1(\vmu_{|(a, \hat{a}_2)}) = u_1(a', \hat{a}_2) + U_1(\vmu_{|(a', \hat{a}_2)}) = U_1(\vmu)$. By the induction hypothesis, for any $a\in S_{\hat{\sigma}_1}$, $U_1(\vmu_{|(a, \hat{a}_2)}) = \sum_{k=1}^{K-1} c_k$ for some $c_k\in V_1^{m,p}$, $k=1,\dots,K-1$. Therefore, $U_1(\vmu) = u_1(a, \hat{a}_2) + U_1(\vmu_{|(a, \hat{a}_2)}) = u_1(\hat{\sigma}_1, \hat{a}_2) + U_1(\vmu_{|(a, \hat{a}_2)}) = \sum_{k=1}^{K} c_k$ for some $c_k\in V_1^{m,p}$, $k=1,\dots,K$. This completes the induction step. 
\end{proof}

\begin{proof}[Proof of \Cref{thm:pure-vs-mixed}]

First we show the condition is sufficient, by showing if the condition is satisfied, we can construct some $T$ and some SPE where \phenom{} occurs. If the condition is satisfied, at least one of (1),(2),(3) must be satisfied. We consider each case here.

If (1) is satisfied, we can use the same construction that is used for the proofs of the pure strategy case and the general case (see the proofs of \Cref{thm:pure,thm:general}, the parts that handle the case where (1) is satisfied). If (3) is satisfied, we can use the same construction that is used for the proof of the pure strategy case (see the proofs of \Cref{thm:pure}, the parts that handle the case where (2) is satisfied).

If (2) is satisfied, then there exists $\hat{\sigma}_1 \in \Delta A_1, a_1' \in A_1, \hat{a}_2 \in A_2$ where (a), (b), and (c) are satisfied. If $|S_{\hat{\sigma}_1}|=1$, i.e., $\hat{\sigma}_1$ is a pure strategy, we can use the same construction that is used for the proof of the pure strategy case (see the proofs of \Cref{thm:pure}, the parts that handle the case where (2) is satisfied). If $|S_{\hat{\sigma}_1}|>1$, i.e., $\hat{\sigma}_1$ is not a pure strategy, we know from (c) that there exists an action $a\in S_{\hat{\sigma}_1}$ such that, for every $a'\in S_{\hat{\sigma}_1}\setminus a$, there exists some integer $n_{a'}\geq 0$ and $d_k^{a'}\in D$, $k=1,\dots,n_{a'}$ such that $u_1(a, \hat{a}_2) - u_1(a', \hat{a}_2) = \sum_{k=1}^{n_{a'}} d_k^{a'}$. Let $a$, $n_{a'}$ for every $a'\in S_{\hat{\sigma}_1}\setminus a$, $d_k^{a'}\in D$ for every $a'\in S_{\hat{\sigma}_1}\setminus a$ and $k=1,\dots,n_{a'}$ be such a set of assignments. We consider a game with $T= 1 + \sum_{a'\in S_{\hat{\sigma}_1}\setminus a} n_{a'} + n_{\bot}$ where $n_{\bot}$ is a large enough integer.
We construct an SPE $\vmu^*$ consisting of segments. Denote $\vmu^t$ as all the $t$-th round behavior strategy profiles in $\vmu^*$ and $\vmu^{t_1:t_2}$ as all the behavior strategy profiles between the $t_1$-th round and the $t_2$-th round in $\vmu^*$. $\vmu^*$ is divided into segments: $\vmu^1, \vmu^{2:T_1}, \vmu^{T_1+1:T_2},\dots, \vmu^{T_{|S_{\hat{\sigma}_1}|-1}+1:T_{|S_{\hat{\sigma}_1}|}}$. Each segment ending at $T_i$ for $i=1,\dots,|S_{\hat{\sigma}_1}|-1$ corresponds to one of $a'\in S_{\hat{\sigma}_1}\setminus a$. We denote the segment corresponding to $a'$ as $\vmu^{a'}$ for every $a'\in S_{\hat{\sigma}_1}\setminus a$; $\vmu^{a'}$ has $n_{a'}$ rounds. We denote the final segment as $\vmu^{\bot}$, which has $n_{\bot}$ rounds. Each $\vmu^{a'}$ for every $a'\in S_{\hat{\sigma}_1}\setminus a$ and $\vmu^{\bot}$ only depend on the play in the first round, not depending on any later rounds. 
we construct $\vmu^*$ as follows:
\begin{itemize}
    \item In the first round, play $\vmu^1 = (\hat{\sigma}_1, \hat{a}_2)$.
    \item For segment $\vmu^{a'}$ for each $a'\in S_{\hat{\sigma}_1}\setminus a$, if player 1's first round play is $a'$, play the sequence of stage-game NEs $\Sigma_1 = (\vsigma_1^1,\dots, \vsigma_1^{n_{a'}})$; otherwise, play the sequence of stage-game NEs $\Sigma_0 = (\vsigma_0^1,\dots, \vsigma_0^{n_{a'}})$. The sequences $\Sigma_1$ and $\Sigma_0$ are chosen according to $u_1(\vsigma_1^k) - u_1(\vsigma_0^k) = d_k^{a'}$ for all $k=1,\dots,n_{a'}$.
    \item Let $\vsmin, \vsmax\in \Nash^{m,p}(G)$ such that $u_1(\vsmin) = \min (V_1^{m,p})$ and $u_1(\vsmax) = \max (V_1^{m,p})$, so $u_1(\vsmax) > u_1(\vsmin)$. For segment $\vmu^{\bot}$, if player 1's first round play is some $a\in S_{\hat{\sigma}_1}$, play $(\vsmax, \vsmax,\dots)$; if player 1's first round play is some $a\notin S_{\hat{\sigma}_1}$, play $(\vsmin, \vsmin,\dots)$.
\end{itemize}
This construction ensures that: 1) player 1 choosing any action $a\in S_{\hat{\sigma}_1}$ in the first round results in the same total payoff in $G(T)$, and 2) player 1 choosing any action $a\notin S_{\hat{\sigma}_1}$ in the first round results in a lower total payoff in $G(T)$ as long as $n_{\bot}$ is large enough. This ensures that such $\vmu^*$ is an SPE of $G(T)$ and the first round play does not form a stage-game NE.
\\
\\
To prove this condition is necessary, we prove that if the condition is not satisfied, then for any $T$ and any SPE $\vmu$ of $G(T)$, \phenom{} does not occur, i.e., the strategy profile at each round must form an NE of the stage game. The condition is not satisfied means all of (1),(2),(3) are false. This can be divided into the following disjoint cases.

\textbf{1. $|V_1^{m,p}| = 0$, $|V_2^{m,p}|=0$.} Here, $\Nash^{m,p}(G) = \emptyset$, so there is no SPE in the repeated game $G(T)$ for any $T$. Therefore, \phenom{} can never occur.

\textbf{2. $|V_1^{m,p}| = 1$, $|V_2^{m,p}|=1$.} 
Using backward induction, we know that in any SPE, the strategy profile at each round must form an NE of the stage game.

\textbf{3. $|V_1^{m,p}| > 1$, $|V_2^{m,p}|>1$.}
Since (1) is false, there does not exist $\hat{\sigma}_1\in \Delta A_1, \hat{a}_2\in A_2$ where $(\hat{\sigma}_1, \hat{a}_2)\notin \Nash^{m,p}(G)$. Therefore, it trivially follows that in any SPE of $G(T)$, the strategy profile at each round must form an NE of the stage game.

\textbf{4. $|V_1^{m,p}| = 1$, $|V_2^{m,p}|>1$. }
Since (3) is false, there does not exist $\hat{\sigma}_1 \in \Delta A_1, \hat{a}_2, a_2' \in A_2$ where $u_2(\hat{\sigma}_1, \hat{a}_2) < u_2(\hat{\sigma}_1, a_2')$ and $\hat{\sigma}_1$ is a best response to $\hat{a}_2$. This means that for any $\hat{\sigma}_1 \in \Delta A_1, \hat{a}_2\in A_2$ where $\hat{\sigma}_1$ is a best response to $\hat{a}_2$, $\hat{a}_2$ is a best response to $\hat{\sigma}_1$, and thus $(\hat{\sigma}_1, \hat{a}_2) \in \Nash^{m,p}(G)$. Now we can use the same backward induction argument as in the proof of \Cref{thm:pure} (the part that proves the condition is necessary, case 4) to prove that in any SPE, the strategy profile at each round must form an NE of the stage game. 

\textbf{5. $|V_1^{m,p}| > 1$, $|V_2^{m,p}| = 1$. } 
Since (2) is false, we know that:
\begin{itemize}
    \item[($\ast$)] For all pure strategy $a_1\in A_1$ and $a_2\in A_2$, if $a_2$ is a best response to $a_1$, then $a_1$ is also a best response to $a_2$, i.e., $(a_1, a_2)\in \Nash^{m,p}(G)$.
    \item[($\ast\ast$)] For all $\sigma_1\in \Delta A_1$ and $a_2\in A_2$ where $a_2$ is a best response to $\sigma_1$ and $\sigma_1$ is not a best response to $a_2$ and $\sigma_1$ is not a pure strategy, denote the set of possible differences in $u_1$ between pairs of NEs in the stage game as $D = \condSet{u_1(\vsigma) - u_1(\vsigma')}{\vsigma,\vsigma' \in \Nash^{m,p}(G)}$, then there does not exist an action from the support of $\sigma_1$ $a\in S_{\sigma_1}$ such that, for every $a'\in S_{\sigma_1}\setminus a$, there exists some integer $n_{a'}\geq 0$ and $d_k^{a'}\in D$, $k=1,\dots,n_{a'}$ such that $u_1(a, a_2) - u_1(a', a_2) = \sum_{k=1}^{n_{a'}} d_k^{a'}$.
\end{itemize}

Now we can use backward induction to prove that in any SPE, the strategy profile at each round must form an NE of the stage game. The strategy profiles in the last round must form stage-game NEs. Given that the strategy profiles in the last $k$ rounds must all form stage-game NEs, consider the $(k+1)$-to-last round. Denote the strategies played in this round as $(\hat{\sigma}_1, \hat{a}_2)$. $\hat{a}_2$ must be a best response to $\hat{\sigma}_1$, since player 2's play in this round does not affect the total payoff they get in the final $k$ rounds. If $\hat{\sigma}_1$ is a pure strategy, then according to ($\ast$), $(\hat{\sigma}_1, \hat{a}_2)$ forms a stage-game NE, which completes the induction step. 

If $\hat{\sigma}_1$ is a mixed strategy (more than one support), assume on the contrary that the strategy profile at this round does not form a stage-game NE, then $\hat{\sigma}_1$ is not a best response to $\hat{a}_2$. By ($\ast\ast$), there does not exist an action from the support of $\hat{\sigma}_1$ $a\in S_{\hat{\sigma}_1}$ such that, for every $a'\in S_{\hat{\sigma}_1}\setminus a$, there exists some integer $n_{a'}\geq 0$ and $d_k^{a'}\in D$, $k=1,\dots,n_{a'}$ such that $u_1(a, \hat{a}_2) - u_1(a', \hat{a}_2) = \sum_{k=1}^{n_{a'}} d_k^{a'}$. Denote $\vmu_{|(a_1, a_2)}^{-k:}$ as the strategy profile in the last $k$ rounds given players played $(a_1, a_2)$ in the first round. For $(\hat{\sigma}_1, \hat{a}_2)$ to be part of a Nash equilibrium of the $(k+1)$-round repeated game, we have for all $a,a'\in S_{\hat{\sigma}_1}$, $u_1(a, \hat{a}_2) + U_1(\vmu_{|(a, \hat{a}_2)}^{-k:}) = u_1(a', \hat{a}_2) + U_1(\vmu_{|(a', \hat{a}_2)}^{-k:})$. By \Cref{lemma:spe_payoff}, for each $a\in S_{\hat{\sigma}_1}$, $U_1(\vmu_{|(a, \hat{a}_2)}^{-k:}) = \sum_{t=1}^k c_t^a$ for some $c_t^a\in V_1^{m,p}$, $t=1,\dots,k$. Therefore, taking an arbitrary $a\in S_{\hat{\sigma}_1}$, for every $a'\in S_{\hat{\sigma}_1}\setminus a$, $u_1(a, \hat{a}_2) - u_1(a', \hat{a}_2) = U_1(\vmu_{|(a', \hat{a}_2)}^{-k:}) - U_1(\vmu_{|(a, \hat{a}_2)}^{-k:}) = \sum_{t=1}^k c_t^{a'} - \sum_{t=1}^k c_t^a = \sum_{t=1}^k d_t^{a'}$, where $d_t^{a'} = c_t^{a'} - c_t^a \in D$. This produces a contradiction. Therefore, the strategy profile at the $(k+1)$-to-last round also forms a stage-game NE. This completes the induction step.

\end{proof}

Again, we can obtain the following corollary regarding the value of $T$ above which \phenom{} can occur if the condition in \Cref{thm:pure-vs-mixed} is satisfied:

\begin{corollary}
For 2-player games where player 1 can use mixed strategies and player 2 can only use pure strategies, given a stage game $G$:
\begin{enumerate}
    \item If $|V_1^{m,p}| > 1$, $|V_2^{m,p}| > 1$, and there exists some $\hat{\sigma}_1\in \Delta A_1, \hat{a}_2\in A_2$ where $(\hat{\sigma}_1, \hat{a}_2) \notin \textrm{Nash}(G)$, then by \Cref{lemma:mixed-pure}, there exists $\hat{a}_1\in A_1$ and $\hat{a}_2\in A_2$ where $(\hat{a}_1, \hat{a}_2) \notin \Nash^{m,p}(G)$, for all $T \geq 2\cdot \max\Big( \frac{\delta_1}{\max (V_1^{m,p})-\min (V_1^{m,p})}, \frac{\delta_2}{\max (V_2^{m,p})- \min (V_2^{m,p})} \Big) +1$ where $\delta_1 = \max_{a_1\in A_1} u_1(a_1, \hat{a}_2) - u_1(\hat{a}_1, \hat{a}_2)$ and $\delta_2 = \max_{a_2\in A_2} u_2(\hat{a}_1, a_2) - u_2(\hat{a}_1, \hat{a}_2)$, there exists some SPE of $G(T)$ where \phenom{} occurs.
    \item If $|V_1^{m,p}| > 1$, $|V_2^{m,p}| = 1$, and there exists $\hat{\sigma}_1 \in \Delta A_1, a_1' \in A_1, \hat{a}_2 \in A_2$ where 
    \begin{enumerate}
        \item $u_1(\hat{\sigma}_1, \hat{a}_2) < u_1(a_1', \hat{a}_2)$, and
        \item $\hat{a}_2$ is a best response to $\hat{\sigma}_1$, and
        \item if $\hat{\sigma}_1$ has more than one support (not a pure strategy), denote the set of possible differences in $u_1$ between pairs of NEs in the stage game as $D = \condSet{u_1(\vsigma) - u_1(\vsigma')}{\vsigma,\vsigma' \in \Nash^{m,p}(G)}$, there exists an action from the support of $\hat{\sigma}_1$ $a\in S_{\hat{\sigma}_1}$ such that, for every $a'\in S_{\hat{\sigma}_1}\setminus a$, there exists some integer $n_{a'}\geq 0$ and $d_k^{a'}\in D$, $k=1,\dots,n_{a'}$ such that $u_1(a, \hat{a}_2) - u_1(a', \hat{a}_2) = \sum_{k=1}^{n_{a'}} d_k^{a'}$,
    \end{enumerate}
    then if $\hat{\sigma}_1$ is a pure strategy, i.e., $|S_{\hat{\sigma}_1}|=1$, for all $T\geq \frac{\max_{a_1\in A_1} u_1(a_1, \hat{a}_2) - u_1(\hat{\sigma}_1, \hat{a}_2)}{\max (V_1^{m,p})-\min (V_1^{m,p})} + 1$, there exists some SPE of $G(T)$ where \phenom{} occurs; if $|S_{\hat{\sigma}_1}|>1$, for all $T\geq 2 + \sum_{a'\in S_{\hat{\sigma}_1}\setminus a} n_{a'} + \frac{\max_{i, i'\in A_1} u_1(i, \hat{a}_2) - u_1(i', \hat{a}_2) }{\max(V_1^{m,p}) - \min(V_1^{m,p})}$, where $a$ and $n_{a'}$ for every $a'\in S_{\hat{\sigma}_1}\setminus a$ are given in (c), there exists some SPE of $G(T)$ where \phenom{} occurs.
    
    \item If $|V_1^{m,p}| = 1$, $|V_2^{m,p}| > 1$, and there exists $\hat{\sigma}_1 \in \Delta A_1, \hat{a}_2, a_2' \in A_2$ where $u_2(\hat{\sigma}_1, \hat{a}_2) < u_2(\hat{\sigma}_1, a_2')$ and $\hat{\sigma}_1$ is a best response to $\hat{a}_2$, then for all $T\geq \frac{\max_{a_2\in A_2} u_2(\hat{\sigma}_1, a_2) - u_2(\hat{\sigma}_1, \hat{a}_2)}{\max (V_2^{m,p})-\min (V_2^{m,p})} + 1$, there exists some SPE of $G(T)$ where \phenom{} occurs.
\end{enumerate}
\end{corollary}

\section{Effect of Changing from Pure Strategies to Mixed Strategies on the Emergence of Local Suboptimality}
\label{sec:suboptimal-change}

In \Cref{sec:pure,sec:mixed,sec:pure_vs_mixed}, we established sufficient and necessary conditions on the stage game $G$ for there exists some $T$ and some SPE of $G(T)$ where \phenom{} occurs for 2-player games, for cases where: 1) both players can only use pure strategies, 2) one player can only use pure strategies and the other player can use mixed strategies, and 3) both players can use mixed strategies.
Essentially, we established a complete characterization of $\glspp$, $\glsmp$, and $\glsmm$ (and therefore $\glopp$, $\glomp$, and $\glomm$).
Based on these results, in this section we study the effect of changing from pure strategies to mixed strategies on the emergence of \phenom{}. We aim to answer the following question: under what conditions on the stage game $G$ will allowing players to play mixed strategies change whether \phenom{} can ever occur in some repeated game $G(T)$?

Essentially, we aim to study the relationships between $\glspp$, $\glsmp$, and $\glsmm$. For example, are there stage games $G$ where $G\notin \glspp$ and $G\in \glsmp$, i.e., \phenom{} can never occur when both players can only use pure strategies but can occur when player 1 obtains access to mixed strategies? What is a complete characterization of such stage games? And in the other direction, are there stage games $G$ where $G\in \glspp$ and $G\notin \glsmp$? We also study the corresponding questions for the relationships between $\glsmp$ and $\glsmm$ and between $\glspp$ and $\glsmm$.

For the simplicity of descriptions, we refer to the case where both players can only use pure strategies as the {\em pure-pure} case, the case where player 1 can use mixed strategies and player 2 can only use pure strategies as the {\em mixed-pure} case, and the case where both players can use mixed strategies as the {\em mixed-mixed} case.

\subsection{Changing from Pure Strategies to Mixed Strategies when the Other Player Can Only Use Pure Strategies}
\label{sec:pp-to-mp}

We first analyze the situation for changing from the pure-pure case to the mixed-pure case, i.e., study the relationship between $\glspp$ and $\glsmp$. 

We first establish a useful theorem:

\begin{theorem}
\label{thm:spe-pp-mp}
For all 2-player stage games $G$, for all $T\in \integer^+$, for all $\vmu\in \spe^{p,p}(G,T)$, $\vmu \in \spe^{m,p}(G,T)$.
\end{theorem}

\begin{proof}
Assume in contradiction that there exists some $G^*$, $T^*$, and $\vmu^*\in \spe^{p,p}(G^*,T^*)$ where $\vmu^* \notin \spe^{m,p}(G^*,T^*)$. Let $k^*$ to be the largest $k$ where $\vmu^*_{|h(k)}$ is not an NE of $G^*(T-k)$ for some history $h(k)$ in the mixed-pure case. Then one of the players must be able to unilaterally change their strategy in the first round of $\vmu^*_{|h(k)}$ to obtain a higher total payoff in $G^*(T-k)$. Player 2 cannot do so since they can still only play pure strategies. For player 1, they cannot do so when they can only play pure strategies, which means any alternative actions in the first round of $\mu^*_{1|h(k)}$ cannot lead to a higher total payoff in $G^*(T-k)$. But this means that any alternative mixed strategies in the first round of $\mu^*_{1|h(k)}$ also cannot lead to a higher total payoff in $G^*(T-k)$. This produces a contradiction. Therefore, for all 2-player stage games $G$, for all $T\in \integer^+$, for all $\vmu\in \spe^{p,p}(G,T)$, $\vmu \in \spe^{m,p}(G,T)$.
\end{proof}

Now we present \Cref{thm:pp-to-mp-1,thm:pp-to-mp-2} and \Cref{lemma:no-use-1,lemma:no-use-2}, which together completely characterize the relationship between $\glspp$ and $\glsmp$.

\begin{theorem}
\label{thm:pp-to-mp-1}
$\glspp \subseteq \glsmp$, i.e.,
for 2-player stage games $G$, if there exists some $T$ and some SPE of $G(T)$ where \phenom{} occurs in the pure-pure case, then there exists some $T$ and some SPE of $G(T)$ where \phenom{} occurs in the mixed-pure case.
\end{theorem}

\begin{proof}
If there exists some $T$ and some SPE of $G(T)$ where \phenom{} occurs in the pure-pure case, denote $T^*$ to be such a $T$ and $\vmu^*$ to be such an SPE of $G(T^*)$. By \Cref{thm:spe-pp-mp}, $\vmu^*$ is also an SPE of $G(T^*)$ in the mixed-pure case.  Since any pure strategy profiles not in $\Nash^{p,p}(G)$ are also not in $\Nash^{m,p}(G)$, \phenom{} occurs in $\vmu^*$ in the mixed-pure case. Therefore, there exists some $T$ and some SPE of $G(T)$ where \phenom{} occurs in the mixed-pure case.
\end{proof}

\begin{lemma}
\label{lemma:no-use-1}
For all stage games $G$ where $|V_1^{p,p}| = 1$ and $|V_2^{p,p}| > 1$, if $G\notin \glspp$, then $G\notin \glsmp$.
\end{lemma}

\begin{proof}
If $|V_1^{p,p}| = 1$, $|V_2^{p,p}| > 1$, and $G\notin \glspp$, by \Cref{thm:pure}, for all $a_2\in A_2$, for all $a_1\in A_1$ that is a best response to $a_2$, $a_2$ is also a best response to $a_1$. Then for any $(\sigma_1, a_2)\in\Nash^{m,p}(G)$, for any $a_1\in S_{\sigma_1}$, $a_1$ is a best response to $a_2$, so $a_2$ is a best response to $a_1$, which means $(a_1, a_2)\in\Nash^{p,p}(G)$ and $u_1(a_1, a_2)\in V_1^{p,p}$. And since $u_1(\sigma_1, a_2) = u_1(a_1, a_2)$ for any $a_1\in S_{\sigma_1}$, $u_1(\sigma_1, a_2)\in V_1^{p,p}$. Thus, $V_1^{m,p}\subseteq V_1^{p,p}$. Furthermore, since $\Nash^{p,p}(G)\subseteq \Nash^{m,p}(G)$, $V_1^{p,p}\subseteq V_1^{m,p}$. Therefore, $|V_1^{p,p}| = |V_1^{m,p}|$. So $|V_1^{m,p}| = 1$, $|V_2^{m,p}| > 1$. We argue that there cannot exist $\hat{\sigma}_1 \in \Delta A_1, \hat{a}_2, a_2' \in A_2$ where $u_2(\hat{\sigma}_1, \hat{a}_2) < u_2(\hat{\sigma}_1, a_2')$ and $\hat{\sigma}_1$ is a best response to $\hat{a}_2$. Assuming there exists such $\hat{\sigma}_1$ and $\hat{a}_2$. Then for all $a_1\in S_{\hat{\sigma}_1}$, $a_1$ is a best response to $\hat{a}_2$, which means $\hat{a}_2$ is also a best response to $a_1$. Then $\hat{a}_2$ is also a best response to $\hat{\sigma}_1$, which produces a contradiction. Therefore, by \Cref{thm:pure-vs-mixed}, there does not exist some $T$ and some SPE of $G(T)$ where \phenom{} occurs in the mixed-pure case, so $G\notin \glsmp$.
\end{proof}

\begin{lemma}
\label{lemma:no-use-2}
For all stage games $G$ where $|V_1^{p,p}| > 1$ and $|V_2^{p,p}| > 1$, if $G\in \glsmp$, then $G\in \glspp$.
\end{lemma}

\begin{proof}
If $|V_1^{p,p}| > 1$, $|V_2^{p,p}| > 1$, and $G\in \glsmp$, since $\Nash^{p,p}(G)\subseteq \Nash^{m,p}(G)$, $|V_1^{m,p}| > 1$ and $|V_2^{m,p}| > 1$. Then since there exists some $T$ and some SPE of $G(T)$ where \phenom{} occurs in the mixed-pure case, by \Cref{thm:pure-vs-mixed}, there exists some $\hat{\sigma}_1\in \Delta A_1, \hat{a}_2\in A_2$ where $(\hat{\sigma}_1, \hat{a}_2) \notin \Nash^{m,p}(G)$. Applying \Cref{lemma:mixed-pure}, there exists some $a_1\in A_1$ and $a_2\in A_2$ where $(a_1, a_2) \notin \Nash^{p,p}(G)$. By \Cref{thm:pure}, there exists some $T$ and some SPE of $G(T)$ where \phenom{} occurs in the pure-pure case, so $G\in \glspp$.
\end{proof}

The above two lemmas show that under certain preconditions on $|V_1^{p,p}|$ and $|V_2^{p,p}|$, $\glspp = \glsmp$.
Now the question is whether $\glspp$ always equals $\glsmp$.
Here we show that this is not the case. We present example stage games $G$ where $G\in \glsmp$ but $G\notin \glspp$ for each of the remaining cases regarding the values of $|V_1^{p,p}|$ and $|V_2^{p,p}|$. This shows that $\glspp \neq \glsmp$. Combined with \Cref{thm:pp-to-mp-1}, this shows that $\glspp$ is a proper subset of $\glsmp$.

\begin{example}
\label{example:pp-mp-00}
\Cref{tab:none_to_new_ne} presents an example stage game $G$ where $|V_1^{p,p}| = 0$, $|V_2^{p,p}| = 0$, $G\in \glsmp$, and $G\notin \glspp$. When both players can only use pure strategies, there is no Nash equilibrium, so $|V_1^{p,p}| = 0$ and $|V_2^{p,p}| = 0$. By \Cref{thm:pure}, $G\notin \glspp$. When player 1 can use mixed strategies and player 2 can only use pure strategies, the set of Nash equilibria are: $(\sigma_1, b_2)$ and $(\sigma_1, c_2)$ for all $\sigma_1$ where $0.25\leq \sigma_1(a_1) \leq 0.75$. So $|V_1^{m,p}| > 1$ and $|V_2^{m,p}| = 1$. And $(b_1, a_2)$ is a strategy profile where player 1 does not play a best response and player 2 plays a best response. Therefore, the condition in \Cref{thm:pure-vs-mixed} is satisfied, so $G\in \glsmp$.

\begin{table}[ht]
    \centering
    \begin{tabular}{c|c|c|c|c|}
    & $a_2$ & $b_2$ & $c_2$ & $d_2$ \\
    \midrule
    $a_1$ & (4,0) & (1,3) & (2,3) & (0,4) \\
      \midrule
    $b_1$ & (0,4) & (1,3) & (2,3) & (4,0) \\
       \midrule
    \end{tabular}
    \caption{Example stage game in matrix form, row player is player 1, column player is player 2. When both players can only use pure strategies, there is no Nash equilibrium. When player 1 can use mixed strategies and player 2 can only use pure strategies, there are multiple Nash equilibria, and \phenom{} can occur in repeated games. }
    \label{tab:none_to_new_ne}
\end{table}
\end{example}

\begin{example}
\label{example:pp-mp-11}
\Cref{tab:new_ne} presents an example stage game $G$ where $|V_1^{p,p}| = 1$, $|V_2^{p,p}| = 1$, $G\in \glsmp$, and $G\notin \glspp$. When both players can only use pure strategies, the only Nash equilibria are $(a_1, a_2)$ and $(b_1, c_2)$, so $|V_1^{p,p}| = 1$ and $|V_2^{p,p}| = 1$. By \Cref{thm:pure}, $G\notin \glspp$. When player 1 can use mixed strategies and player 2 can only use pure strategies, $(\sigma_1, b_2)$ for all $\sigma_1$ where $0.25\leq \sigma_1(a_1) \leq 0.75$ are Nash equilibria. So $|V_1^{m,p}| > 1$ and $|V_2^{m,p}| > 1$. And $(b_1, a_2)\notin \Nash^{m,p}(G)$. Therefore, the condition in \Cref{thm:pure-vs-mixed} is satisfied, so $G\in \glsmp$.

\begin{table}[ht]
    \centering
    \begin{tabular}{c|c|c|c|}
    & $a_2$ & $b_2$ & $c_2$ \\
    \midrule
    $a_1$ & (4,4) & (1,3)  & (0,0) \\
      \midrule
    $b_1$ & (0,0) & (1,3) & (4,4) \\
       \midrule
    \end{tabular}
    \caption{Example stage game in matrix form, row player is player 1, column player is player 2. $|V_1^{p,p}| = 1$, $|V_2^{p,p}| = 1$, $|V_1^{m,p}| > 1$, and $|V_2^{m,p}| > 1$. }
    \label{tab:new_ne}
\end{table}
\end{example}

\begin{example}
\label{example:pp-mp-g1}
\Cref{tab:no_best_response} presents an example stage game $G$ where $|V_1^{p,p}| > 1$, $|V_2^{p,p}| = 1$, $G\in \glsmp$, and $G\notin \glspp$. For this game $G$, the set of pure Nash equilibria is: $(a_1, a_2)$, $(b_1, a_2)$, $(b_1, b_2)$, and $(c_1, b_2)$. Therefore, $|V_1^{p,p}| > 1$ and $|V_2^{p,p}|=1$. We can see that for all pure strategy profiles where player 2 plays a best response, player 1 also plays a best response. So the condition in \Cref{thm:pure} is not satisfied, and $G\notin \glspp$.
The strategy profile $(\sigma_1, a_2)$ where $\sigma_1(a_1) = 0.9$ and $\sigma_1(c_1) = 0.1$ is an example strategy profile where player 2 plays a best response and player 1 does not play a best response. And $u_1(a_1, a_2) - u_1(c_1, a_2)$ can be expressed as the sum of some sequence of values in $D = \condSet{u_1(\vsigma) - u_1(\vsigma')}{\vsigma,\vsigma' \in \Nash^{m,p}(G)}$. So the condition in \Cref{thm:pure-vs-mixed} is satisfied, and $G\in \glsmp$. Intuitively speaking, in the pure-pure case, although $|V_1^{p,p}| > 1$ makes player 1 potentially vulnerable to threats, there is no way to construct such a threat since there is no strategy profiles where player 1 does not play a best response but player 2 plays a best response. But in the mixed-pure case, such off-(stage-game)-Nash strategy profiles become available, which makes the threat possible.
\end{example}

Furthermore, the following theorem completely characterizes $\glsmp \setminus \glspp$, i.e., the set of 2-player stage games $G$ where 1) in the pure-pure case, \phenom{} can never occur, and 2) in the mixed-pure case, there exists some $T$ and some SPE of $G(T)$ where \phenom{} occurs.

\begin{theorem}
\label{thm:pp-to-mp-2}
For 2-player stage games $G$, $G\notin \glspp$ and $G\in \glsmp$ if and only if
\begin{enumerate}
    \item $|V_1^{p,p}| = 0$, $|V_2^{p,p}| = 0$, and the condition in \Cref{thm:pure-vs-mixed} is satisfied, OR
    \item $|V_1^{p,p}| = 1$, $|V_2^{p,p}| = 1$, and the condition in \Cref{thm:pure-vs-mixed} is satisfied, OR
    \item $|V_1^{p,p}| > 1$, $|V_2^{p,p}| = 1$, there does not exist $\hat{a}_1, a_1' \in A_1, \hat{a}_2 \in A_2$ where $u_1(\hat{a}_1, \hat{a}_2) < u_1(a_1', \hat{a}_2)$ and $\hat{a}_2$ is a best response to $\hat{a}_1$, and the condition in \Cref{thm:pure-vs-mixed} is satisfied.
\end{enumerate}
\end{theorem}

\begin{proof}
To show the above condition is sufficient, we show that each of 1, 2, and 3 is sufficient. 

\textbf{1.} $|V_1^{p,p}| = 0$ and $|V_2^{p,p}| = 0$ mean that $G$ has no pure Nash equilibrium. So in the pure-pure case, there is no SPE for any repeated game $G(T)$, therefore \phenom{} can never occur. Furthermore, the condition in \Cref{thm:pure-vs-mixed} is satisfied implies that in the mixed-pure case, there exists some $T$ and some SPE of $G(T)$ where \phenom{} occurs. 
\Cref{example:pp-mp-00} presents a example stage game $G$ that satisfies this condition.

\textbf{2.} Since $|V_1^{p,p}| = 1$ and $|V_2^{p,p}| = 1$, by \Cref{thm:pure}, in the pure-pure case, \phenom{} can never occur. And the condition in \Cref{thm:pure-vs-mixed} is satisfied implies that in the mixed-pure case, there exists some $T$ and some SPE of $G(T)$ where \phenom{} occurs. \Cref{example:pp-mp-11} presents an example stage game $G$ that satisfies this condition.

\textbf{3.} Since $|V_1^{p,p}| > 1$, $|V_2^{p,p}| = 1$, and there does not exist $\hat{a}_1, a_1' \in A_1, \hat{a}_2 \in A_2$ where $u_1(\hat{a}_1, \hat{a}_2) < u_1(a_1', \hat{a}_2)$ and $\hat{a}_2$ is a best response to $\hat{a}_1$, by \Cref{thm:pure}, in the pure-pure case, \phenom{} can never occur. And the condition in \Cref{thm:pure-vs-mixed} is satisfied implies that in the mixed-pure case, there exists some $T$ and some SPE of $G(T)$ where \phenom{} occurs. \Cref{example:pp-mp-g1} presents an example stage game $G$ that satisfies this condition. 
\\
\\
To show the condition is necessary, we split all stage games $G$ into five disjoint cases based on $|V_1^{p,p}|$ and $|V_2^{p,p}|$:
\begin{itemize}
    \item[(a)] $|V_1^{p,p}| = 0$ and $|V_2^{p,p}| = 0$,
    \item[(b)] $|V_1^{p,p}| = 1$ and $|V_2^{p,p}| = 1$,
    \item[(c)] $|V_1^{p,p}| > 1$ and $|V_2^{p,p}| = 1$,
    \item[(d)] $|V_1^{p,p}| = 1$ and $|V_2^{p,p}| > 1$,
    \item[(e)] $|V_1^{p,p}| > 1$ and $|V_2^{p,p}| > 1$.
\end{itemize}

For cases (a), (b), and (c), a direct application of \Cref{thm:pure,thm:pure-vs-mixed} implies that our target condition is necessary. For cases (d) and (e), \Cref{lemma:no-use-1,lemma:no-use-2} show that $\glsmp \setminus \glspp$ is empty under these two cases. Therefore, our target condition is necessary.
\end{proof}

\subsection{Changing from Pure Strategies to Mixed Strategies when the Other Player Can Use Mixed Strategies}
\label{sec:mp-to-mm}

Next, we analyze the situation for changing from the mixed-pure case to the mixed-mixed case, i.e., study the relationship between $\glsmp$ and $\glsmm$. 

Following the same argument as the proof of \Cref{thm:spe-pp-mp}, we can prove the following theorem, which is useful for the results in this part.
\begin{theorem}
\label{thm:spe-mp-mm}
For all 2-player stage games $G$, for all $T\in \integer^+$, for all $\vmu\in \spe^{m,p}(G,T)$, $\vmu \in \spe^{m,m}(G,T)$.
\end{theorem}

\begin{proof}
Assume in contradiction that there exists some $G^*$, $T^*$, and $\vmu^*\in \spe^{m,p}(G^*,T^*)$ where $\vmu^* \notin \spe^{m,m}(G^*,T^*)$. Let $k^*$ to be the largest $k$ where $\vmu^*_{|h(k)}$ is not an NE of $G^*(T-k)$ for some history $h(k)$ in the mixed-mixed case. Then one of the players must be able to unilaterally change their strategy in the first round of $\vmu^*_{|h(k)}$ to obtain a higher total payoff in $G^*(T-k)$. Player 1 cannot do so since their strategy space is the same in the mixed-pure case and the mixed-mixed case. For player 2, they cannot do so when they can only play pure strategies, which means any alternative actions in the first round of $\mu^*_{2|h(k)}$ cannot lead to a higher total payoff in $G^*(T-k)$. But this means that any alternative mixed strategies in the first round of $\mu^*_{2|h(k)}$ also cannot lead to a higher total payoff in $G^*(T-k)$. This produces a contradiction. Therefore, for all 2-player stage games $G$, for all $T\in \integer^+$, for all $\vmu\in \spe^{m,p}(G,T)$, $\vmu \in \spe^{m,m}(G,T)$.
\end{proof}

Now we present \Cref{thm:mp-to-mm-1,thm:mp-to-mm-2} and \Cref{lemma:no-use-3}, which together completely characterize the relationship between $\glsmp$ and $\glsmm$.

\begin{theorem}
\label{thm:mp-to-mm-1}
$\glsmp \subseteq \glsmm$, i.e., for 2-player stage games $G$, if there exists some $T$ and some SPE of $G(T)$ where \phenom{} occurs in the mixed-pure case, then there exists some $T$ and some SPE of $G(T)$ where \phenom{} occurs in the mixed-mixed case.
\end{theorem}

\begin{proof}
If there exists some $T$ and some SPE of $G(T)$ where \phenom{} occurs in the mixed-pure case, denote $T^*$ to be such a $T$ and $\vmu^*$ to be such an SPE of $G(T^*)$. By \Cref{thm:spe-mp-mm}, $\vmu^*$ is also an SPE of $G(T^*)$ in the mixed-mixed case. For any strategy profile $(\sigma_1, a_2)$ where $\sigma_1\in \Delta A_1$ and $a_2\in A_2$, if $(\sigma_1, a_2)\notin \Nash^{m,p}(G)$, $(\sigma_1, a_2)\notin \Nash^{m,m}(G)$. So \phenom{} occurs in $\vmu^*$ in the mixed-mixed case. Therefore, there exists some $T$ and some SPE of $G(T)$ where \phenom{} occurs in the mixed-mixed case.
\end{proof}

\begin{lemma}
\label{lemma:no-use-3}
For all stage games $G$ where $|V_1^{m,p}| > 1$ and $|V_2^{m,p}| > 1$, if $G\in \glsmm$, then $G\in \glsmp$.
\end{lemma}

\begin{proof}
If $|V_1^{m,p}| > 1$, $|V_2^{m,p}| > 1$, and $G\in \glsmm$, since $\Nash^{m,p}(G)\subseteq \Nash^{m,m}(G)$, $|V_1^{m,m}| > 1$ and $|V_2^{m,m}| > 1$. Then since there exists some $T$ and some SPE of $G(T)$ where \phenom{} occurs in the mixed-mixed case, by \Cref{thm:general}, there exists some $\hat{\sigma}_1 \in \Delta A_1, \hat{\sigma}_2\in \Delta A_2$ where $(\hat{\sigma}_1, \hat{\sigma}_2) \notin \Nash^{m,m}(G)$. Applying \Cref{lemma:mixed-pure}, there exists some $a_1\in A_1$ and $a_2\in A_2$ where $(a_1, a_2) \notin \Nash^{p,p}(G)$. By \Cref{thm:pure-vs-mixed}, there exists some $T$ and some SPE of $G(T)$ where \phenom{} occurs in the mixed-pure case, so $G\in \glsmp$.
\end{proof}

The above lemma shows that under certain preconditions on $|V_1^{m,p}|$ and $|V_2^{m,p}|$, $\glsmp = \glsmm$.
Now the question is whether $\glsmp$ always equals $\glsmm$.
Here we show that this is not the case. We present example stage games $G$ where $G\in \glsmm$ but $G\notin \glsmp$ for each of the remaining cases regarding the values of $|V_1^{m,p}|$ and $|V_2^{m,p}|$. This shows that $\glsmp \neq \glsmm$. Combined with \Cref{thm:mp-to-mm-1}, this shows that $\glsmp$ is a proper subset of $\glsmm$.

\begin{example}
\label{example:mp-mm-00}
\Cref{tab:none_to_new_ne_flip} presents an example stage game $G$ where $|V_1^{m,p}| = 0$, $|V_2^{m,p}| = 0$, $G\in \glsmm$, and $G\notin \glsmp$. This game is essentially the stage game presented in \Cref{tab:none_to_new_ne} with column player being player 1 and row player being player 2. When player 1 can play mixed strategies and player 2 can only play pure strategies, there is no Nash equilibrium, so $|V_1^{m,p}| = 0$ and $|V_2^{m,p}| = 0$. By \Cref{thm:pure-vs-mixed}, $G\notin \glsmp$. When both players can use mixed strategies, $|V_1^{m,m}| = 1$ and $|V_2^{m,m}| > 1$. And $(a_1, b_2)$ is a strategy profile where player 1 plays a best response and player 2 does not play a best response. Therefore, the condition in \Cref{thm:general} is satisfied, so $G\in \glsmm$.

\begin{table}[ht]
    \centering
    \begin{tabular}{c|c|c|}
    & $a_2$ & $b_2$ \\
    \midrule
    $a_1$ & (0,4) & (4,0) \\
    \midrule
    $b_1$ & (3,1) & (3,1) \\
    \midrule
    $c_1$ & (3,2) & (3,2) \\
    \midrule
    $d_1$ & (4,0) & (0,4) \\
   \midrule
    \end{tabular}
    \caption{Example stage game in matrix form, row player is player 1, column player is player 2. When player 1 can play mixed strategies and player 2 can only play pure strategies, there is no Nash equilibrium. When both players can use mixed strategies, there are multiple Nash equilibria, and \phenom{} can occur in repeated games. }
    \label{tab:none_to_new_ne_flip}
\end{table}
\end{example}

\begin{example}
\label{example:mp-mm-11}
\Cref{tab:new_ne_flip} presents an example stage game $G$ where $|V_1^{m,p}| = 1$, $|V_2^{m,p}| = 1$, $G\in \glsmm$, and $G\notin \glsmp$. This game is essentially the stage game presented in \Cref{tab:new_ne} with column player being player 1 and row player being player 2. Following the arguments in \Cref{example:pp-mp-11}, for this game, $|V_1^{m,p}| = 1$, $|V_2^{m,p}| = 1$, $|V_1^{m,m}| > 1$, $|V_2^{m,m}| > 1$, and $(a_1, b_2)\notin \Nash^{m,m}(G)$. Therefore, by \Cref{thm:pure-vs-mixed,thm:general}, $G\in \glsmm$ and $G\notin \glsmp$.
\begin{table}[ht]
    \centering
    \begin{tabular}{c|c|c|}
    & $a_2$ & $b_2$ \\
    \midrule
    $a_1$ & (4,4) & (0,0) \\
      \midrule
    $b_1$ & (3,1) & (3,1) \\
       \midrule
   $c_1$ & (0,0) & (4,4) \\
   \midrule
    \end{tabular}
    \caption{Example stage game in matrix form, row player is player 1, column player is player 2. $|V_1^{m,p}| = 1$, $|V_2^{m,p}| = 1$, $|V_1^{m,m}| > 1$, and $|V_2^{m,m}| > 1$. }
    \label{tab:new_ne_flip}
\end{table}
\end{example}

\begin{example}
\label{example:mp-mm-g1}
\Cref{tab:no_multiset_sum} presents an example stage game $G$ where $|V_1^{m,p}| > 1$, $|V_2^{m,p}| = 1$, $G\in \glsmm$, and $G\notin \glsmp$. When player 1 can play mixed strategies and player 2 can only play pure strategies, the set of Nash equilibria is: 1) $(\sigma_{\lambda}, a_2)$ for all $0\leq \lambda\leq 1$ where $\sigma_{\lambda}(a_1)=\lambda$ and $\sigma_{\lambda}(b_1) = 1 - \lambda$, and 2) $(\sigma'_{\theta}, b_2)$ for all $0\leq \theta\leq 1$ where $\sigma'_{\theta}(b_1)=\theta$ and $\sigma'_{\theta}(c_1) = 1 - \lambda$. So $|V_1^{m,p}| > 1$ and $|V_2^{m,p}| = 1$. For any $\hat{\sigma}_1 \in \Delta A_1, \hat{a}_2 \in A_2$ where $\hat{a}_2$ is a best response to $\hat{\sigma}_1$ and $\hat{\sigma}_1$ is not a best response to $\hat{a}_2$, if $\hat{a}_2=a_2$, $S_{\hat{\sigma}_1}$ must include $c_1$ and at least one of $a_1$ and $b_1$; if $\hat{a}_2=b_2$, $S_{\hat{\sigma}_1}$ must include $a_1$ and at least one of $b_1$ and $c_1$. In all these cases, there is some $a, a'\in S_{\hat{\sigma}_1}$ where $u_1(a, \hat{a}_2) - u_1(a', \hat{a}_2)=0.5$. But $V_1 = \{2,3\}$ only contains integer elements. So such $u_1(a, \hat{a}_2) - u_1(a', \hat{a}_2)$ can never be expressed as the sum of a sequence of values from $D = \condSet{u_1(\vsigma) - u_1(\vsigma')}{\vsigma,\vsigma' \in \Nash^{m,p}(G)}$. Therefore, the condition in \Cref{thm:pure-vs-mixed} is not satisfied, whereas the condition in \Cref{thm:general} is satisfied. So $G\in \glsmm$ and $G\notin \glsmp$. Intuitively speaking, in the mixed-pure case, although $|V_1^{m,p}| > 1$ makes player 1 potentially vulnerable to threats, there is no way to construct such a threat with the available strategy profiles. But in the mixed-mixed case, as player 2 obtains access to mixed strategies, new mechanisms for constructing threats become available (as are used in the proof of \Cref{thm:general}), which enables \phenom{} to occur. 

\begin{table}[ht]
    \centering
    \begin{tabular}{c|c|c|}
    & $a_2$ & $b_2$ \\
    \midrule
     $a_1$ & (3,2)   & (1.5,1) \\
      \midrule
     $b_1$ & (3,2)   & (2,2) \\
       \midrule
     $c_1$ & (2.5,1)  & (2,2) \\
       \midrule
    \end{tabular}
    \caption{Example stage game in matrix form, row player is player 1, column player is player 2. $|V_1^{m,p}| > 1$, $|V_2^{m,p}| = 1$. \Phenom{} can never occur in the mixed-pure case, but can occur in the mixed-mixed case.}
    \label{tab:no_multiset_sum}
\end{table}
\end{example}

\begin{example}
\label{example:mp-mm-1g}
\Cref{tab:no_best_response_flip} presents an example stage game $G$ where $|V_1^{m,p}| = 1$, $|V_2^{m,p}| > 1$, $G\in \glsmm$, and $G\notin \glsmp$.
This game is essentially the stage game presented in \Cref{tab:no_best_response} with column player being player 1 and row player being player 2. The same discussion in \Cref{example:pp-mp-g1} applies here.

\begin{table}[ht]
    \centering
    \begin{tabular}{c|c|c|c|}
    & $a_2$ & $b_2$ & $c_2$ \\
    \midrule
    $a_1$ & (2,3) & (2,3) & (1,1) \\
      \midrule
    $b_1$ & (1,1) & (2,2) & (2,2) \\
       \midrule
    \end{tabular}
    \caption{Example stage game in matrix form, row player is player 1, column player is player 2. For all $a\in A_2$, for all $\sigma_1\in \Delta A_1$ that is a best response to $a$, $a$ is also a best response to $\sigma_1$. }
    \label{tab:no_best_response_flip}
\end{table}
\end{example}

Furthermore, the following theorem completely characterizes $\glsmm \setminus \glsmp$, i.e., the set of 2-player stage games $G$ where 1) in the mixed-pure case, \phenom{} can never occur, and 2) in the mixed-mixed case, there exists some $T$ and some SPE of $G(T)$ where \phenom{} occurs.

\begin{theorem}
\label{thm:mp-to-mm-2}
For 2-player stage games $G$, $G\notin \glsmp$ and $G\in \glsmm$ if and only if
\begin{enumerate}
    \item $|V_1^{m,p}| = 0$, $|V_2^{m,p}| = 0$, and the condition in \Cref{thm:general} is satisfied, OR
    \item $|V_1^{m,p}| = 1$, $|V_2^{m,p}| = 1$, and the condition in \Cref{thm:general} is satisfied, OR
    \item $|V_1^{m,p}| > 1$, $|V_2^{m,p}| = 1$, the condition in \Cref{thm:pure-vs-mixed} is not satisfied, and the condition in \Cref{thm:general} is satisfied, OR
    \item $|V_1^{m,p}| = 1$, $|V_2^{m,p}| > 1$, the condition in \Cref{thm:pure-vs-mixed} is not satisfied, and the condition in \Cref{thm:general} is satisfied.
\end{enumerate}
\end{theorem}

\begin{proof}
Same as in the proof of \Cref{thm:pp-to-mp-2}, a direct application of \Cref{thm:pure-vs-mixed,thm:general} shows that the above condition is sufficient. \Cref{example:mp-mm-00,example:mp-mm-11,example:mp-mm-1g,example:mp-mm-g1} present example stage games $G$ that belong to each of the cases 1, 2, 3, and 4.

To show the condition is necessary, we split all stage games $G$ into five disjoint cases based on $|V_1^{m,p}|$ and $|V_2^{m,p}|$:
\begin{itemize}
    \item[(a)] $|V_1^{m,p}| = 0$ and $|V_2^{m,p}| = 0$,
    \item[(b)] $|V_1^{m,p}| = 1$ and $|V_2^{m,p}| = 1$,
    \item[(c)] $|V_1^{m,p}| > 1$ and $|V_2^{m,p}| = 1$,
    \item[(d)] $|V_1^{m,p}| = 1$ and $|V_2^{m,p}| > 1$,
    \item[(e)] $|V_1^{m,p}| > 1$ and $|V_2^{m,p}| > 1$.
\end{itemize}

For cases (a), (b), (c), and (d), a direct application of \Cref{thm:pure-vs-mixed,thm:general} implies that our target condition is necessary. For case (e), \Cref{lemma:no-use-3} shows that $\glsmm \setminus \glsmp$ is empty under this case. Therefore, our target condition is necessary.
\end{proof}

\subsection{Changing Both Players from Pure Strategies to Mixed Strategies}
\label{sec:pp-to-mm}

Finally, we analyze the situation for changing from the pure-pure case to the mixed-mixed case, i.e., i.e., study the relationship between $\glspp$ and $\glsmm$. The results and example games for this situation can mostly be derived directly from the results in \Cref{sec:pp-to-mp,sec:mp-to-mm}.

\begin{theorem}
\label{thm:pp-to-mm-1}
$\glspp \subseteq \glsmm$, i.e., for 2-player stage games $G$, if there exists some $T$ and some SPE of $G(T)$ where \phenom{} occurs in the pure-pure case, then there exists some $T$ and some SPE of $G(T)$ where \phenom{} occurs in the mixed-mixed case.
\end{theorem}

\begin{proof}
\Cref{thm:pp-to-mp-1,thm:mp-to-mm-1} directly imply this result.
\end{proof}

\begin{lemma}
\label{lemma:no-use-4}
For all stage games $G$ where $|V_1^{p,p}| > 1$ and $|V_2^{p,p}| > 1$, if $G\in \glsmm$, then $G\in \glspp$.
\end{lemma}

\begin{proof}
If $|V_1^{p,p}| > 1$, $|V_2^{p,p}| > 1$, since $\Nash^{p,p}(G)\subseteq \Nash^{m,p}(G)\subseteq \Nash^{m,m}(G)$, $|V_1^{m,p}| > 1$, $|V_2^{m,p}| > 1$, $|V_1^{m,m}| > 1$ and $|V_2^{m,m}| > 1$. Therefore, by applying \Cref{lemma:no-use-2,lemma:no-use-3}, we have if $G\in \glsmm$, then $G\in \glspp$.
\end{proof}

The above lemma shows that under certain preconditions on $|V_1^{p,p}|$ and $|V_2^{p,p}|$, $\glspp = \glsmm$.
Here we present example stage games $G$ where $G\in \glsmm$ but $G\notin \glspp$ for each of the remaining cases regarding the values of $|V_1^{p,p}|$ and $|V_2^{p,p}|$. These examples are reused from \Cref{sec:pp-to-mp,sec:mp-to-mm}. This shows that $\glspp \neq \glsmm$. Combined with \Cref{thm:pp-to-mm-1}, this shows that $\glspp$ is a proper subset of $\glsmm$.

\begin{example}
\label{example:pp-mm-00}
\Cref{tab:none_to_new_ne} presents an example stage game $G$ where $|V_1^{p,p}| = 0$, $|V_2^{p,p}| = 0$, $G\in \glsmm$, and $G\notin \glspp$.
\end{example}

\begin{example}
\label{example:pp-mm-11}
\Cref{tab:new_ne} presents an example stage game $G$ where $|V_1^{p,p}| = 1$, $|V_2^{p,p}| = 1$, $G\in \glsmm$, and $G\notin \glspp$.
\end{example}

\begin{example}
\label{example:pp-mm-g1}
\Cref{tab:no_best_response} presents an example stage game $G$ where $|V_1^{p,p}| > 1$, $|V_2^{p,p}| = 1$, $G\in \glsmm$, and $G\notin \glspp$.
\end{example}

\begin{example}
\label{example:pp-mm-1g}
\Cref{tab:no_best_response_flip} presents an example stage game $G$ where $|V_1^{p,p}| = 1$, $|V_2^{p,p}| > 1$, $G\in \glsmm$, and $G\notin \glspp$.
\end{example}

The following theorem completely characterizes $\glsmm \setminus \glspp$, i.e., the set of 2-player stage games $G$ where 1) in the pure-pure case, \phenom{} can never occur, and 2) in the mixed-mixed case, there exists some $T$ and some SPE of $G(T)$ where \phenom{} occurs.

\begin{theorem}
\label{thm:pp-to-mm-2}
For 2-player stage games $G$, $G\notin \glspp$ and $G\in \glsmm$ if and only if
\begin{enumerate}
    \item $|V_1^{p,p}| = 0$, $|V_2^{p,p}| = 0$, and the condition in \Cref{thm:general} is satisfied, OR
    \item $|V_1^{p,p}| = 1$, $|V_2^{p,p}| = 1$, and the condition in \Cref{thm:general} is satisfied, OR
    \item $|V_1^{p,p}| > 1$, $|V_2^{p,p}| = 1$, the condition in \Cref{thm:pure} is not satisfied, and the condition in \Cref{thm:general} is satisfied, OR
    \item $|V_1^{p,p}| = 1$, $|V_2^{p,p}| > 1$, the condition in \Cref{thm:pure} is not satisfied, and the condition in \Cref{thm:general} is satisfied.
\end{enumerate}
\end{theorem}

\begin{proof}
Same as in the proofs of \Cref{thm:pp-to-mp-2,thm:mp-to-mm-2}, a direct application of \Cref{thm:pure,thm:general} shows that the above condition is sufficient. \Cref{example:pp-mm-00,example:pp-mm-11,example:pp-mm-1g,example:pp-mm-g1} present example stage games $G$ that belong to each of the cases 1, 2, 3, and 4.

To show the condition is necessary, again, we split all stage games $G$ into five disjoint cases based on $|V_1^{p,p}|$ and $|V_2^{p,p}|$:
\begin{itemize}
    \item[(a)] $|V_1^{p,p}| = 0$ and $|V_2^{p,p}| = 0$,
    \item[(b)] $|V_1^{p,p}| = 1$ and $|V_2^{p,p}| = 1$,
    \item[(c)] $|V_1^{p,p}| > 1$ and $|V_2^{p,p}| = 1$,
    \item[(d)] $|V_1^{p,p}| = 1$ and $|V_2^{p,p}| > 1$,
    \item[(e)] $|V_1^{p,p}| > 1$ and $|V_2^{p,p}| > 1$.
\end{itemize}
For cases (a), (b), (c), and (d), a direct application of \Cref{thm:pure,thm:general} implies that our target condition is necessary.
For case (e), \Cref{lemma:no-use-4} shows that $\glsmm \setminus \glsmp$ is empty under this case. Therefore, our target condition is necessary.
\end{proof}

\subsection{Discussion}

The central question we focus on in this section is how changing a player (or both players) from pure-strategies-only to mixed-strategies-allowed can affect the emergence of \phenom{}.
We present here an intuitive interpretation of the results established in this section.

First, if \phenom{} can occur before the change, then after changing any player (or both players) from pure-strategies-only to mixed-strategies-allowed, \phenom{} can still occur (\Cref{thm:pp-to-mp-1,thm:mp-to-mm-1,thm:pp-to-mm-1}). So allowing players to play mixed strategies can never prohibit the emergence of \phenom{}.

On the other hand, it is possible that \phenom{} can never occur before the change, but after changing one player (or both players) from pure-strategies-only to mixed-strategies-allowed, \phenom{} can occur. Such phenomena can happen through two different mechanisms. The first one is through the introduction of new stage-game Nash equilibria. Before the change, there might be no stage-game NE, or there is only one payoff value attainable at stage-game NEs for each player ($V_1=V_2=1$), which makes neither of the players vulnerable to potential threats that force them to play locally suboptimally. After allowing one (or both) player(s) to play mixed strategies, a new set of stage-game NEs becomes available. This makes some $|V_i| > 1$, making that player vulnerable to potential threats. Cases 1 and 2 in \Cref{thm:pp-to-mp-2,thm:mp-to-mm-2,thm:pp-to-mm-2} and \Cref{example:pp-mp-00,example:pp-mp-11,example:mp-mm-00,example:mp-mm-11,example:pp-mm-00,example:pp-mm-11} belong to this type.

For the second mechanism, before the change, some player already have $|V_i| > 1$, which means they are potentially vulnerable to threats. However, there is no way of constructing such a threat in any SPE given the available strategy profiles, so \phenom{} cannot occur. After allowing one (or both) player(s) to play mixed strategies, with the newly available strategy profiles, it becomes possible to construct such a threat, which makes it possible for \phenom{} to occur. We show that this can happen in the following cases:
\begin{itemize}
    \item The player that is potentially vulnerable to threats before the change (i.e., $|V_i|>1$) changes from pure-strategies-only to mixed-strategies-allowed. This change can open up vulnerabilities for themselves, regardless of whether the opponent has access to mixed strategies or not. Case 3 in \Cref{thm:pp-to-mp-2} ($|V_1^{p,p}| > 1$, $|V_2^{p,p}| = 1$, player 1 changes from pure to mixed, player 2 can only play pure), case 4 in \Cref{thm:mp-to-mm-2} ($|V_1^{m,p}| = 1$, $|V_2^{m,p}| > 1$, player 2 changes from pure to mixed, player 1 can play mixed), and the corresponding example games \Cref{example:pp-mp-g1,example:mp-mm-1g} demonstrate this type.
    \item A player changes from pure-strategies-only to mixed-strategies-allowed when their opponent is potentially vulnerable to threats before the change (i.e., $|V_i|>1$). Here, only if their opponent has access to mixed strategies will such change be useful to create threats that were not possible before the change. Case 3 in \Cref{thm:mp-to-mm-2} and the corresponding example game \Cref{example:mp-mm-g1} demonstrate this situation. Importantly, if their opponent only has access to pure strategies, changing from pure-strategies-only to mixed-strategies-allowed will not enable a player to construct threats if it was impossible to create threats before the change. This is shown in \Cref{lemma:no-use-1}.
\end{itemize}

\section{Computational Aspects}
\label{sec:comp}

In this section, we consider the computational aspect of the problem: given an arbitrary 2-player stage game $G$, how to (algorithmically) decide if there exists some $T$ and some SPE of $G(T)$ where \phenom{} occurs. We focus on the general case where mixed strategies are allowed. A naive approach is to enumerate over $T$ and solve for all subgame-perfect equilibria for each $G(T)$. Such an approach is not only computationally inefficient, but also not guaranteed to terminate due to the unboundedness of $T$. This leaves open the question of whether the above problem is decidable or not.

\Cref{thm:general} proves a necessary and sufficient condition that is solely described on the stage game $G$, independent of $T$. Based on this condition, we present here a more efficient algorithm for deciding the above problem (\Cref{alg:decide}). This algorithm also shows that the above problem is decidable.

\begin{algorithm}
\small
\caption{Given stage game $G$, decide if there exists some $T$ and some SPE of $G(T)$ where \phenom{} occurs}
\label{alg:decide}
\begin{algorithmic}[1]
\Require $G = \{u_1(i,j), u_2(i,j)\}_{i\in [|A_1|], j\in[|A_2|]}$
\Ensure True/False
     \Function{DecideLocalSubOptimality}{$G$}
          \State $\Nash(G) \gets$ \Call{FindAllNash}{$G$} 
          \State {uniqueV1, uniqueV2 $\gets$ \Call{IsValueUnique}{$\Nash(G), G$}}
          \If{uniqueV1}
            \If{uniqueV2}
                \State \textbf{return} False
            \Else
                \State \textbf{return} \Call{ExistOff1Best}{$G$}
            \EndIf
          \Else
            \If{uniqueV2}
                \State \textbf{return} \Call{ExistOff2Best}{$G$}
            \Else
                \State \textbf{return} \Call{ExistOff}{$G$}
            \EndIf
          \EndIf
     \EndFunction
\end{algorithmic}
\end{algorithm}

The input to the algorithm is the stage game $G$, represented as a matrix of payoffs for every action profile $\{u_1(i,j), u_2(i,j)\}_{i\in [|A_1|], j\in[|A_2|]}$. Overall, the algorithm consists of three steps. The first step is to compute the set of all Nash equilibria of $G$. The second step is to determine if the set of payoff values attainable at $\Nash(G)$ is unique for each player, so as to know which case of the condition in \Cref{thm:general} we need to further check. The third step is to check if there exists the respective off-Nash strategy profiles required for each case. We now present the algorithm for each of the three steps. 

\subsection{Compute All Nash Equilibria}
There are many existing algorithms for computing the set of all Nash equilibria of two-player normal form games \cite{dickhaut1993program, audet2001enumeration, von2002computing, Avis2010}. In general, any existing algorithm can be used here as long as it can handle degenerate games where there are an infinite number of Nash equilibria. An example of such algorithm can be found in \cite{Avis2010}. It handles the potentially infinite number of Nash equilibria in degenerate games by computing the finite set of \textit{extreme equilibria}; the set of all Nash equilibria is then completely described by polytopes obtained from subsets of the extreme equilibria. 

\paragraph{Complexity}
In general, the problem of finding all Nash equilibria of a two-player normal form game is NP-hard (as \cite{GILBOA1989} shows that deciding if a game has a unique Nash equilibrium is NP-hard). Therefore, any algorithm for \textsc{FindAllNash} takes exponential time in the worst case (unless P=NP).

\subsection{Determine the Uniqueness of Payoffs at Equilibrium}
\textsc{IsValueUnique} returns two Boolean values; the first (resp. second) return value is \texttt{True} if $|V_1|=1$ (resp. $|V_2|=1$). This function can be achieved by evaluating payoffs at each extreme equilibrium and compare to see if there are more than one values for each player.

\paragraph{Complexity}
In general, a non-degenerate 2-player bimatrix game can have an exponential number of Nash equilibria \cite{von1999new, quint1997theorem}. So \textsc{IsValueUnique} can take exponential time. But in practice, this step can be done in the first step (\textsc{FindAllNash}) with a constant factor overhead, by evaluating the payoff of each extreme equilibrium immediately after the extreme equilibrium is computed in \textsc{FindAllNash} and comparing with the payoffs of the previous equiliria.

\subsection{Check the Existence of Required Off-Nash Strategy Profiles}
Based on the uniqueness of payoffs attainable at equilibrium for each player (\texttt{uniqueV1} and \texttt{uniqueV2}), we need to check the existence of off-Nash strategy profiles with the requirements corresponding to each case. 

If $|V_1| > 1$ and $|V_2| > 1$ (both \texttt{uniqueV1} and \texttt{uniqueV2} are \texttt{False}), we simply need to check the existence of an off-Nash strategy profile without further requirements. By \Cref{lemma:mixed-pure}, it suffices to check the existence of an off-Nash \textit{pure} strategy profile. \Cref{alg:existoff} achieves this functionality. 

\begin{algorithm}
\small
\caption{Check the existence of an off-Nash strategy profile}
\label{alg:existoff}
\begin{algorithmic}[1]
\Require $G = \{u_1(i,j), u_2(i,j)\}_{i\in [|A_1|], j\in[|A_2|]}$
\Ensure True/False
\Function{ExistOff}{$G$}
    \For{$j = 1, \dots, |A_2|$}
        \If{not all $u_1(i,j)$ for $i= 1, \dots, |A_1|$ are the same}
            \State \textbf{return} True
        \EndIf
    \EndFor
    \For{$i = 1, \dots, |A_1|$}
        \If{not all $u_2(i,j)$ for $j= 1, \dots, |A_2|$ are the same}
            \State \textbf{return} True
        \EndIf
    \EndFor
    \State \textbf{return} False
\EndFunction
\end{algorithmic}
\end{algorithm}

If $|V_1| > 1$ and $|V_2| = 1$ (\texttt{uniqueV1} is \texttt{False} and \texttt{uniqueV2} is \texttt{True}), we need to check the existence of an off-Nash strategy profile where player 2 plays a best response. The following lemma proves that it suffices to check the existence of an off-Nash strategy profile where player 2 plays a \textit{pure strategy} best response.

\begin{lemma}
\label{lemma:mixed-pure-2best}
For any two-player game $G$, if there exists $\sigma_1\in\Delta A_1$, $\sigma_2\in\Delta A_2$, $a_1'\in A_1$ where $u_1(\sigma_1, \sigma_2) < u_1(a_1', \sigma_2)$ and $\sigma_2$ is a best response to $\sigma_1$, then there exists $\sigma_1\in\Delta A_1$, $a_2\in A_2$, $a_1'\in A_1$ where $u_1(\sigma_1, a_2) < u_1(a_1', a_2)$ and $a_2$ is a best response to $\sigma_1$.
\end{lemma}

\begin{proof}
Let $S_{\sigma_2}$ be the support of $\sigma_2$. $\sigma_2$ is a best response to $\sigma_1$ implies that any $a\in S_{\sigma_2}$ is a best response to $\sigma_1$. We have $u_1(a_1', \sigma_2) - u_1(\sigma_1, \sigma_2) = \sum_{a\in S_{\sigma_2}} \sigma_2(a)\cdot \Big(u_1(a_1', a) - u_1(\sigma_1, a) \Big)$. Since $u_1(a_1', \sigma_2) - u_1(\sigma_1, \sigma_2) > 0$, there exists some $a_2 \in S_{\sigma_2}$ where $u_1(a_1', a_2) - u_1(\sigma_1, a_2) > 0$. This $\sigma_1\in\Delta A_1$, $a_2\in A_2$, $a_1'\in A_1$ satisfies $u_1(\sigma_1, a_2) < u_1(a_1', a_2)$ and $a_2$ is a best response to $\sigma_1$.
\end{proof}
\Cref{alg:existoff2best} presents a method for checking the existence of an off-Nash strategy profile where player 2 plays a pure strategy best response. The idea is as follows. For each possible pure strategy $j$ of player 2, we construct linear programs with constraints on player 1's mixed strategy (represented by probabilities $\{x_{i'}\}_{i'=1}^{|A_1|}$) such that $j$ is a best response. We aim to find for every action $i$ of player 1 that is not a best response to $j$, if $j$ can be a best response to a mixed strategy of player 1 containing $i$ in its support. The presented linear program achieves this purpose. If we can find such a mixed strategy $\sigma_1$ for player 1, then $(\sigma_1, j)$ is an instance of an off-Nash strategy profile where player 2 plays a best response, as desired. Since we exhaustively enumerate over all possible cases, this method is complete.

\textsc{ExistOff1Best} can be implemented using the same algorithm, exchanging player 1 and 2.

\begin{algorithm}
\small
\caption{Check the existence of an off-Nash strategy profile where player 2 plays a best response}
\label{alg:existoff2best}
\begin{algorithmic}[1]
\Require $G = \{u_1(i,j), u_2(i,j)\}_{i\in [|A_1|], j\in[|A_2|]}$
\Ensure True/False
\Function{ExistOff2Best}{$G$}
    \For{$j=1,\dots,|A_2|$}
        \State $c \gets \max_i u_1(i,j)$
        \For{$i\in [|A_1|]$ where $u_1(i,j)< c$}
            \State max\_xi $\gets$ solve the following linear program
            \begin{align*}
                \textrm{maximize: } & x_i \\
                \textrm{subject to: } & x_{i'} \geq 0, i'=1,\dots,|A_1| \\
                & \sum_{i'=1}^{|A_1|} x_{i'} = 1 \\
                & \sum_{i'=1}^{|A_1|} x_{i'} \cdot u_2(i',j) \geq \sum_{i'=1}^{|A_1|} x_{i'} \cdot u_2(i',j'), j'=1,\dots,|A_2| \\
            \end{align*}
            \If{max\_xi > 0}
                \State \textbf{return} True
            \EndIf
        \EndFor
    \EndFor
    \State \textbf{return} False
\EndFunction
\end{algorithmic}
\end{algorithm}

\paragraph{Complexity}
\textsc{ExistOff} has complexity $\mathcal{O}(|A_1|\cdot |A_2|)$. \textsc{ExistOff2Best} (and similarly \textsc{ExistOff1Best}) involves solving $\mathcal{O}(|A_1|\cdot |A_2|)$ instances of polynomial-sized linear programs. Since a linear program can be solved in polynomial time, $\textsc{ExistOff2Best}$ is a polynomial time algorithm. A vanilla support enumeration algorithm (such as \cite{dickhaut1993program}), which enumerates over all possible supports (subsets of the action sets) for the mixed strategies $\sigma_1$ and $\sigma_2$ in the required strategy profile and solve for each case, requires exponential time since there is an exponential number of possible supports. The algorithm we present here is more efficient.
\newline
\newline
Overall, the computational bottleneck is step 1, since finding all Nash equilibria is NP-hard. Step 2 can be computed within step 1 with a constant factor overhead. Step 3 can be computed in polynomial time.

\section{Generalization to $n$-player Games}
\label{sec:n-player}

Denote $I(G) = \condSet{i}{|V_i| = 1}$ as the set of players that have a unique payoff attainable at $\Nash(G)$. Given any strategy profile $\vsigma$, denote $B(\vsigma) = \condSet{i}{\sigma_i \textrm{ is a best response to } \vsigma_{-i}}$ as the set of players that plays a best response strategy. 
\begin{theorem}[$n$-player, sufficient condition]
\label{thm:n-player-suf}
    For general $n$-player games (mixed strategies allowed), a sufficient condition on the stage game $G$ for there exists some $T$ and some SPE of $G(T)$ where \phenom{} occurs is:

    There exists a strategy profile $\hat{\vsigma} = (\hat{\sigma}_1, \dots, \hat{\sigma}_n)$ where $I(G)\subseteq B(\hat{\vsigma})$ and $B(\hat{\vsigma}) \neq [n]$, and 
    \begin{enumerate}
        \item there exists $\vsigma, \vsigma' \in \Nash(G)$ where 
            \begin{enumerate}
                \item for all $\lambda\in [0,1]$, $\lambda\vsigma + (1- \lambda)\vsigma' \in \Nash(G)$, and 
                \item for some $i\in I(G)$, $\sigma_i\neq \sigma'_i$, 
            \end{enumerate}
            OR
        \item for all $i\in[n]\setminus B(\hat{\vsigma})$, 
        \begin{enumerate}
            \item $|S_{\hat{\sigma}_i}| = 1$, i.e. $\hat{\sigma}_i$ is a pure strategy, or
            \item $V_i$ contains a non-zero length continuous interval, or
            \item denote the set of possible differences in $u_i$ between pairs of NEs in the stage game as $D_i = \condSet{u_i(\vsigma) - u_i(\vsigma')}{\vsigma,\vsigma' \in \Nash(G)}$, there exists an action from the support of $\hat{\sigma}_i$ $a\in S_{\hat{\sigma}_i}$ such that, for every $a'\in S_{\hat{\sigma}_i}\setminus a$, there exists some integers $n_{\va_{I(G)}}\geq 0$ and $d_k^{\va_{I(G)}}\in D_i, k=1,\dots,n_{\va_{I(G)}}$ for each $\va_{I(G)}\in \times_{i\in I(G)} S_{\hat{\sigma}_i}$ such that $u_i(a', \hat{\vsigma}_{-i}) - u_i(a, \hat{\vsigma}_{-i}) = \sum_{\va_{I(G)}\in \times_{i\in I(G)} S_{\hat{\sigma}_i}} \hat{\vsigma}(\va_{I(G)})\cdot \sum_{k=1}^{n_{\va_{I(G)}}} d_k^{\va_{I(G)}}$. 
        \end{enumerate}
    \end{enumerate}
\end{theorem}

\begin{proof}

We prove the condition is sufficient by showing if the condition is satisfied, we can always construct some $T$ and some SPE where \phenom{} occurs.

We construct an SPE $\vmu^*$ consisting of segments. WLOG, let the first $m$ players be the ones that do not play their best response in $\hat{\vsigma}$, i.e. $[m] = [n]\setminus B(\hat{\vsigma})$. Denote $\vmu^t$ as all the $t$-th round behavior strategy profiles in $\vmu^*$ and $\vmu^{t_1:t_2}$ as all the behavior strategy profiles between the $t_1$-th round and the $t_2$-th round in $\vmu^*$. $\vmu^*$ is divided into segments: $\vmu^1, \vmu^{2:T_1}, \vmu^{T_1+1:T_2},\dots, \vmu^{T_{m-1}+1:T_m}$. We set $\vmu^1 = \hat{\vsigma}$. We let the strategies in each segment only depend on the play in the first round, not depending on any other segments. So given a play in the first round, each $\vmu^{T_{i-1}+1:T_i}$ is an SPE of the $(T_i - T_{i-1})$-round subgame.
\\
\\
\textbf{(1) is satisfied.} 
We let the strategies in the segment ending at $T_i$ only depend on the action played by player $i$ in the first round. In the following, we denote $\vmu^{T_{i-1}+1:T_i}_{|a_i}$ as the strategy profile in the segment ending at $T_i$ given player $i$ plays $a_i$ in the first round.

$\vmu^{T_{i-1}+1:T_i}$ is constructed as follows. Pick $a_i^m \in S_{\hat{\sigma}_i}$ such that $u_i(a_i^m, \hat{\vsigma}_{-i}) = \max_{a_i \in S_{\hat{\sigma}_i}} u_i(a_i, \hat{\vsigma}_{-i})$. We construct $\vmu^{T_{i-1}+1:T_i}$ such that for all $a'_i\in S_{\hat{\sigma}_i}$, $U_i(\vmu^{T_{i-1}+1:T_i}_{|a'_i}) - U_i(\vmu^{T_{i-1}+1:T_i}_{|a_i^m}) = u_i(a_i^m, \hat{\vsigma}_{-i}) - u_i(a'_i, \hat{\vsigma}_{-i})$. Let $\vsmin, \vsmax\in \Nash(G)$ such that $u_i(\vsmin) = \min (V_i)$ and $u_i(\vsmax) = \max (V_i)$, so $u_i(\vsmax) > u_i(\vsmin)$. 
By using $\vsigma, \vsigma'$ given in (1), we can construct SPEs that achieve a continuous range of values of $U_i$ for all $i\notin I(G)$. Denote $j\in I(G)$ such that $\sigma_j\neq \sigma'_j$, and $a_j \in A_j$ such that $\sigma_j(a_j)\neq \sigma'_j(a_j)$. WLOG, let $\sigma_j(a_j) > \sigma'_j(a_j)$. We can construct SPEs $\vmu(\lambda)$ parameterized by $\lambda$ as:
\begin{itemize}
    \item In the first round, play $\lambda\vsigma + (1- \lambda)\vsigma'$.
    \item If the first round play by player $j$ is $a_j$, players play $\vsmax$ in all later rounds; otherwise, players play $\vsmin$ in all later rounds.
\end{itemize}
By varying $\lambda$ from 0 to 1, we can obtain a continuous range of values for $U_i(\vmu(\lambda))$. And by setting the number of rounds larger, the value range can be arbitrarily large. Then we can use $\vmu(\lambda)$ for $\vmu^{T_{i-1}+1:T_i}_{|a_i}$ for each $a_i\in S_{\hat{\sigma}_i}$ with separately and appropriately assigned $\lambda$'s, such that $U_i(\vmu^{T_{i-1}+1:T_i}_{|a_i}) + u_i(a_i, \hat{\vsigma}_{-i})$ is a constant across all $a_i\in S_{\hat{\sigma}_i}$.
Furthermore, we can include a segment in $\vmu^{T_{i-1}+1:T_i}$ where if some $a_i\in S_{\hat{\sigma}_i}$ is played by player $i$ in the first round, players play according to $\vsmax$; otherwise, players play according to $\vsmin$.

The above construction ensures that $\vmu^*$ is an SPE. And the first round strategy profile in $\vmu^*$ does not form a stage game Nash equilibrium.
\\
\\
\textbf{(2) is satisfied.} 
We construct $\vmu^{T_{i-1}+1:T_i}$ based on which case of (a), (b), and (c) is satisfied. Again, let $\vsmin, \vsmax\in \Nash(G)$ such that $u_i(\vsmin) = \min (V_i)$ and $u_i(\vsmax) = \max (V_i)$.

If (a) is satisfied, i.e. $\hat{\sigma}_i$ is a pure strategy, denote $a_i$ as the support for $\hat{\sigma}_i$. $\vmu^{T_{i-1}+1:T_i}$ is then constructed as: if player $i$ plays $a_i$ in the first round of $\vmu^*$, players play according to $\vsmax$ in all rounds in $\vmu^{T_{i-1}+1:T_i}$; otherwise, players play according to $\vsmin$ in all rounds in $\vmu^{T_{i-1}+1:T_i}$.

If (b) is satisfied, i.e. $V_i$ contains a non-zero length continuous interval, we can use a similar construction as the case when (1) is satisfied, replacing $\vmu(\lambda)$ with repetitions of stage game NE that achieves appropriate values of $V_i$ in the continuous interval, such that $U_i(\vmu^{T_{i-1}+1:T_i}_{|a_i}) + u_i(a_i, \hat{\vsigma}_{-i})$ is a constant across all $a_i\in S_{\hat{\sigma}_i}$.

If (c) is satisfied, we let $\vmu^{T_{i-1}+1:T_i}$ depend on the play of the set of players $i\cup I(G)$ in the first round. Taking $a$ as the chosen action from the support of $\hat{\sigma}_i$ in condition (c). We further divide $\vmu^{T_{i-1}+1:T_i}$ into $|S_{\hat{\sigma}_i}|-1$ segments and denote $\vmu^{T_{i-1}+1:T_i}[a']$ as the segment corresponding to $a'\in S_{\hat{\sigma}_i}\setminus a$. For each $\va_{I(G)}\in \times_{i\in I(G)} S_{\hat{\sigma}_i}$ and $a'\in S_{\hat{\sigma}_i}\setminus a$, we construct the segment $\vmu^{T_{i-1}+1:T_i}[a']$ by setting $\vmu^{T_{i-1}+1:T_i}_{|(a, \va_{I(G)})}[a']$ and $\vmu^{T_{i-1}+1:T_i}_{|(a', \va_{I(G)})}[a']$ using corresponding sequences of stage game NEs as specified by $\{d_k^{\va_{I(G)}}\}_{k=1}^{n_{\va_{I(G)}}}$ given in (c). $\vmu^{T_{i-1}+1:T_i}_{|(a'', \va_{I(G)})}[a'] = \vmu^{T_{i-1}+1:T_i}_{|(a, \va_{I(G)})}[a']$ for $a'' \in S_{\hat{\sigma}_i}\setminus \{a'\}$. This construction ensures that $U_i(\vmu^{T_{i-1}+1:T_i}_{|a_i}) + u_i(a_i, \hat{\vsigma}_{-i})$ is a constant across all $a_i\in S_{\hat{\sigma}_i}$.

The above construction ensures that $\vmu^*$ is an SPE. And the first round strategy profile in $\vmu^*$ does not form a stage game Nash equilibrium.
    
\end{proof}

\begin{theorem}[$n$-player, necessary condition]
\label{thm:n-player-nec}
    For general $n$-player games (mixed strategies allowed), a necessary condition on the stage game $G$ for there exists some $T$ and some SPE of $G(T)$ where \phenom{} occurs is: 

    There exists a strategy profile $\hat{\vsigma} = (\hat{\sigma}_1, \dots, \hat{\sigma}_n)$ where $I(G)\subseteq B(\vsigma)$ and $B(\vsigma) \neq [n]$.
\end{theorem}

\begin{proof}

We prove the above condition is necessary by showing that if there exists some $T$ and some SPE of $G(T)$ where \phenom{} occurs, there exists a strategy profile $\hat{\vsigma} = (\hat{\sigma}_1, \dots, \hat{\sigma}_n)$ where $I(G)\subseteq B(\vsigma)$ and $B(\vsigma) \neq [n]$.

Let $T^*$ and some SPE $\vmu^*$ of $G(T^*)$ to be an instance where \phenom{} occurs. Let $k^* = \max \condSet{k}{\exists h(k) \textrm{ s.t. } \vmu(h(k))\notin \Nash(G)}$ be the last round where off-Nash play occurs and let $\vmu(h^*(k^*)) \notin \Nash(G)$ be a behavior strategy profile that does not form a stage-game Nash equilibrium. Denote $G_{|h^*(k^*)}$ as the subgame starting from $h^*(k^*)$. Consider $\vmu_{|h^*(k^*)}$, the strategy profile in the subgame $G_{|h^*(k^*)}$. By the above construction, all the behavior strategy profiles in $\vmu_{|h^*(k^*)}$ after the first round belong to $\Nash(G)$. Therefore, for every player $i\in I(G)$, their total payoff in $G_{|h^*(k^*)}$ after the first round does not depend on what is played in the first round. So they must play their best responses in the first round, i.e. in $\vmu(h^*(k^*))$. So $I(G)\subseteq B(\vmu(h^*(k^*)))$. And since $\vmu(h^*(k^*))\notin\Nash(G)$, $B(\vmu(h^*(k^*))) \neq [n]$. So $\vmu(h^*(k^*))$ is a strategy profile that satisfies our target condition. 

\end{proof}

\paragraph{Remark} 
In the 2-player case, we are able to prove some properties that hold for 2-player games (\Cref{lemma:matrix} and the subsequent arguments in the proof of \Cref{thm:general} that uses \Cref{lemma:matrix} to show there exists a connected component of off-Nash strategy profiles), which allows the proof of the sufficient and necessary condition for the general case where mixed strategies are allowed. It is not clear whether similar properties hold for $n$-player games. Therefore, the questions of 1) what is a sufficient and necessary condition for $n$-player games, and 2) is \Cref{ques:comp} decidable for $n$-player games, remain open problems.

\section{Related Work}
\label{sec:suboptimal-related}

%\paragraph{Folk Theorem}
Under the theme of analyzing equilibrium solutions in repeated games, a large body of work focuses on Folk Theorems, where the property of interest is: all feasible and individually rational payoff profiles can be attained in equilibria of the repeated game. 
In the context of infinitely repeated games, the original Folk Theorem asserts that all feasible and individually rational (see \Cref{sec:example-diff} for the definitions) payoff profiles can be attained in Nash equilibria of infinitely repeated games with sufficiently little discounting. 
This result is widely known in the field but not formally published, which is why it is called Folk Theorem.
\cite{aumann1994long,rubinstein1979equilibrium} show that the same result holds when we consider subgame-perfect equilibria and assume no discounting. \cite{fudenberg1986folk} proves a sufficient condition for Folk Theorem for subgame-perfect equilibria in infinitely repeated games with discounting. \cite{friedman1971non,friedman1982oligopoly} consider a variation of Folk Theorem where they show that any feasible payoff profile that Pareto dominates a Nash equilibrium of the stage game can be attained in a subgame-perfect equilibrium of the infinitely repeated game with discounting.
In the context of finitely repeated games, \cite{Benoit1985} obtained sufficient conditions for Folk Theorem for subgame-perfect equilibria, and later \cite{Smith1995} establishes necessary and sufficient conditions for subgame-perfect equilibria. Both results rely on mixed strategies are observable, meaning that players can directly observe the mixed strategies (i.e., probability distributions) used by other players in previous rounds of the game, not just the realized actions in the previous rounds; \cite{gossner1995} establishes sufficient conditions for subgame-perfect equilibrium without this assumption. 
\cite{benoit1987} obtained sufficient conditions for Nash equilibria, and \cite{gonzalez2006finitely} establishes sufficient and necessary conditions for Nash equilibria. Folk Theorem has also been studied in a broader class of repeated game models. \cite{fudenberg1986folk} considers Folk Theorem for finitely repeated game with incomplete information. \cite{fudenberg1994folk} considers infinitely repeated game with imperfect monitoring. \cite{dutta1995folk} considers infinite horizon stochastic games with perfect monitoring, and later \cite{fudenberg2011folk} considers infinite horizon stochastic games with imperfect monitoring.

A major difference between the above line of work and this work is that Folk Theorems consider the set of payoffs attainable, whereas this work considers the occurrence of off-(stage-game)-Nash play. As we demonstrate in \Cref{sec:example-diff}, the property considered in Folk Theorems and the \phenom{} property considered in this work do not have direct implications in either direction. Therefore, unlike this research, none of the above research establishes a sufficient and necessary condition for off-(stage-game)-Nash play to occur in finitely repeated games.

%\paragraph{Equilibrium value set in repeated games}
Several works in the literature establish additional characterizations on the equilibrium value set in repeated games. When the preconditions of Folk Theorems do not hold, these results provide some characterizations on the equilibrium value set. 
\cite{demeze2020complete} provides a complete characterization of the set of pure strategy SPE payoff profiles in the limit as the time horizon increases for finitely repeated games with perfect monitoring. 
\cite{radner1985repeated,radner1986example} characterize limiting behavior of the equilibrium value set of infinitely repeated games with imperfect monitoring as the discount factor approaches 1.
\cite{abreu1990toward} further proves properties of the equilibrium value set in infinitely repeated games with discounting and imperfect monitoring. Again, this line of work considers the set of payoffs attainable, whereas our work considers the occurrence of off-(stage-game)-Nash play. Unlike our research, none of the above research establishes a sufficient and necessary condition for off-(stage-game)-Nash play to occur in finitely repeated games.

%\paragraph{Equilibrium analyses in repeated games}
%In a broader context, \cite{aumann1995repeated,mertens2015repeated} are two books that present comprehensive studies on equilibrium solutions in repeated games. To the best of our knowledge, the research presented in this thesis is the first to establish a complete characterization of when off-(stage-game)-Nash play can occur in two-player finitely repeated games.

%Besides equilibrium value set, there is research that analyzes other aspects of equilibrium solutions in repeated games. \cite{kreps1982rational} shows that cooperation occurs in sequential equilibria of finitely repeated prisoner's dilemma when there is certain incomplete information on one of the players. \cite{rubinstein1986finite} shows that cooperation cannot occur in infinitely repeated prisoner's dilemma when players use finite automata as their strategies. \cite{fudenberg1983subgame} shows that SPEs of infinitely repeated games arise in the limit from $\epsilon$-SPEs of finitely repeated games.

\section*{Acknowledgment}

We thank Kai Jia and Idan Orzech for the helpful discussion and proofreading. 

\newpage
\bibliographystyle{splncs04nat}
\bibliography{refs}

\end{document}